\documentclass[11pt, epsfig]{article}
\usepackage{epsfig, amsmath, amssymb, amsthm, times}
\usepackage{graphicx}
\usepackage{color}
\usepackage{bm} 
\usepackage{stackrel}

\newcommand{\R}{\mathbb{R}}
\newcommand{\T}{\mathbb{T}}
\newcommand{\C}{\mathbb{C}}
\newcommand{\Z}{\mathbb{Z}}
\newcommand{\N}{\mathbb{N}}

\newcommand{\cB}{\mathcal{B}}
\newcommand{\cS}{\mathcal{S}}

\newcommand{\cP}{\mathcal{P}}

\def\lbl{\label}
\def\be{\begin{equation}}
\def\ee{\end{equation}}
\def\p{\partial}

\newcommand{\iu}{{i\mkern1mu}}

\def\1{\mathbf{1}}

\newcommand{\cE}{\mathcal{E}}

\newcommand{\cyc}[1]{\operatorname{Cyc}\left(#1\right)}

\newtheorem{theorem}{Theorem}[section]

\newtheorem{lemma}[theorem]{Lemma}

\newtheorem{remark}[theorem]{Remark}
\newtheorem{definition}[theorem]{Definition}
\newtheorem{assumption}[theorem]{Assumption}

\title{Kuramoto model on Sierpinski Gasket I: Harmonic maps}
\author{Georgi S. Medvedev\thanks{Department of Mathematics, Drexel University,
		{\tt medvedev@drexel.edu}} \and Mathew S. Mizuhara\thanks{Department of 
		Mathematics and Statistics,
		The College of New Jersey,
		{\tt  mizuharm@tcnj.edu}}}

\begin{document}

\maketitle

\begin{abstract}
  Motivated by the study of attractors in the Kuramoto model (KM) on graphs approximating the
  Sierpinski gasket (SG), we revisit the problem of harmonic maps (HMs) from SG to the circle,
  first considered by Strichartz. We provide a geometric proof of Strichartz’s theorem,
  which states that for a prescribed degree and suitable boundary conditions, there exists a
  unique HM from SG to the circle.
  Furthermore, we extend this result to HMs on post-critically finite (p.c.f.) fractals.

  For continuous functions on SG, we define a degree given by vector of integers of arbitrary finite
  length. We show that the degree determines a homotopy class of a continuous function on SG with
  values in the circle. This provides an analog of the Hopf degree theorem for continuous functions on SG.

  We then move on to analyze the HMs on SG. At the heart of our method lies an original construction
  of the covering spaces for the SG. After lifting continuous functions on the SG with values in the unit circle to 
  continuous real-valued functions on the covering space, we use the harmonic
  extension algorithm to obtain a harmonic function on the covering space.
 The desired HM is obtained by restricting the domain of the resultant harmonic function to the
 fundamental domain and projecting the range to the circle.  Each covering space is  constructed
 separately for HMs of a given homotopy class, capturing its intrinsic topology.
 
 We show that with suitable modifications the method applies to p.c.f. fractals, a large class of self-similar domains.
 We illustrate our method of constructing the HMs using numerical examples of HMs from the SG
 to the circle and discuss the construction of the covering spaces for several representative p.c.f. fractals, 
 including the $3$-level SG, the hexagasket, and the pentagasket.

 The results of this paper provide the foundation for the follow-up work where we give
 a complete description of the attractors in the KM on graphs approximating p.c.f. fractals.
 Specifically, we show that all HMs identified in this paper are stable steady states of the KM. \\

 \noindent \noindent{ \small{\it Keywords.}
   Coupled dynamical system, fractal, Sierpinski gasket, harmonic extension, covering space.}

 \noindent\noindent{\small{\it MSC2020.} 34C15, 58E20, 14F40.}
\end{abstract}

\vfill\newpage

\section{Introduction}
\setcounter{equation}{0}

  The Kuramoto model (KM) of coupled phase oscillators provides a framework for the analysis
  of collective dynamics in coupled systems \cite{Kur84, KurPikRos}. Owing to its analytical
  tractability and broad range of applications, the KM has become a widely used model for studying
  synchronization in large ensembles of oscillators. Applications include phase locking and rhythm
  generation in neuronal networks, synchronous flashing in firefly populations \cite{Str-Sync},
  synchronization of Josephson junction arrays, and control of power networks \cite{DorBul12}.

  The KM is best known for the universal scenario for the transition from incoherence to
  synchronization in large populations of coupled oscillators with random intrinsic frequencies
  \cite{Kur75, StrMir91, Str00, Chi15}.
  In prior work with Chiba, we analyzed this transition for the KM on general networks.
  Our analysis characterized the precise effect of network organization on the onset of synchrony
  and on the spatial structure of the resulting synchronized states
  \cite{ChiMed19a, ChiMed19b, ChiMed22, CMM18, CMM23},
  as well as on the structure of partially coherent and chimera states emerging
  when the coherent state loses stability \cite{MM22, CMM2022}.

  In the present work, we initiate a systematic study of the KM on
  self-similar networks (see \cite{Med2026} for related work for nonlocally coupled networks).
  Many real-world networks exhibit hierarchical organization
  across multiple scales. This feature is prominent in the synaptic organization of the mammalian brain,
  in particular in cortical networks, where connectivity spans several structural levels, from individual neurons
  to microcircuits, to multilayered columns with intra- and inter-columnar connections, and further
  to large-scale interregional networks \cite{Shepard}. Such hierarchical structure supports increasingly
  abstract and integrated information processing \cite{Harris, Tso}. Hierarchical organization is also reflected
  in artificial neural networks and in large technological networks such as the World Wide Web \cite{DillKum02}.
  Direct modeling of the connectivity in these systems is analytically intractable. It is therefore natural to
  seek graph models that capture multiscale organization while remaining amenable to rigorous analysis.
  To this end, we propose to use graphs approximating fractals as models of hierarchical networks.
  In the continuum limit, the KM on such graphs reduces to the heat equation on a fractal,
  providing a setting for rigorous analysis.


We now present the model that will be analyzed in this work. Consider the KM on a sequence of graphs
$\Gamma_n$ approximating a fractal set. To fix the ideas,
as such a set we take SG. SG is defined as a unique attractor of the system of contracting
similarities:
\be\lbl{FixP}
G=\bigcup_{i=1}^3 F_i(G),
\ee
where
$$
F_i(x)=\frac{1}{2} \left( x-v_i\right)+v_i,\quad i\in [3]\doteq\{1,2,3\},
$$
and $v_i$'s are vertices of an equilateral triangle (see Fig.~\ref{fig:sg}) \cite{Falc-FracGeom}. 
$\Gamma_n$ is a graph on $\frac{3(3^n+1)}{2}$ nodes, $n\in\N$.
The set of nodes of $\Gamma_n$ is denoted by $V_n$.
Two nodes $x,y\in V_n$
are adjacent (denoted $x\sim_n y$) if they belong  to the
same $n$-cell, i.e., $x,y \in T_w:=F_w(T)$ for some $w=(w_1,w_2,\dots, w_n)\in \cS^n$. Here,
$\cS$ stands for the alphabet of three symbols $\{1,2,3\}$, $T$ is an equilateral triangle with
vertices $v_1, v_2,$ and $v_3$ and $F_w$ abbreviates
$F_{w_1}\circ F_{w_2}\circ\dots\circ F_{w_n}$. Geometrically, $\Gamma_n=\bigcup_{|w|=n} \partial T_w$
(see Fig.~\ref{fig:sg_graph}).
The set $V_0 = \{v_1,v_2,v_3\}$ will be used to assign the boundary values for 
the Laplace equation on $G$.
\begin{figure}[h]
	\centering
	\includegraphics[width = .35\textwidth]{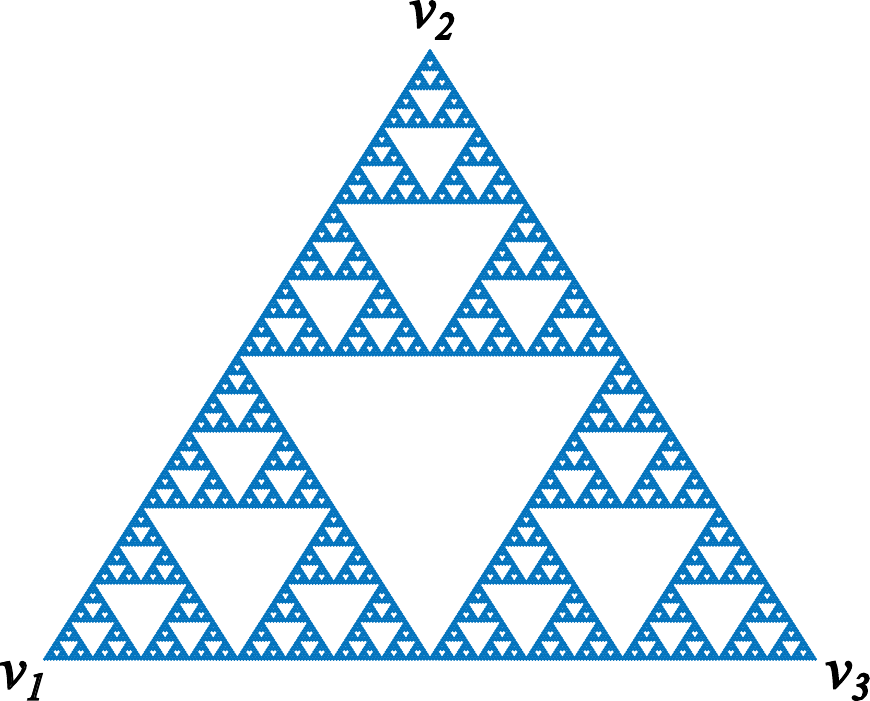}
	\caption{The Sierpinski Gasket.}
	\label{fig:sg}
\end{figure}
 
We use $\Gamma_n$ as a model of a self-similar network. The KM on $\Gamma_n$ has the following
form
\begin{equation}\label{KM}
  \dot u_i =\left(\frac{5}{3}\right)^n \sum_{j:\,j\sim_n i}
  \sin\left( 2\pi \left(u_j- u_i\right)\right),
  \end{equation}
  where $u_i:=u(t,i)$ is a phase oscillator located at node
  $i\in V_n\setminus V_0$. $u_i$ takes values in the unit  circle $\T:=\R/\Z$. The boundary conditions
  imposed at $V_0$ will be discussed below.
  
\begin{figure}[h]
	\centering
	\includegraphics[width=.8\textwidth]{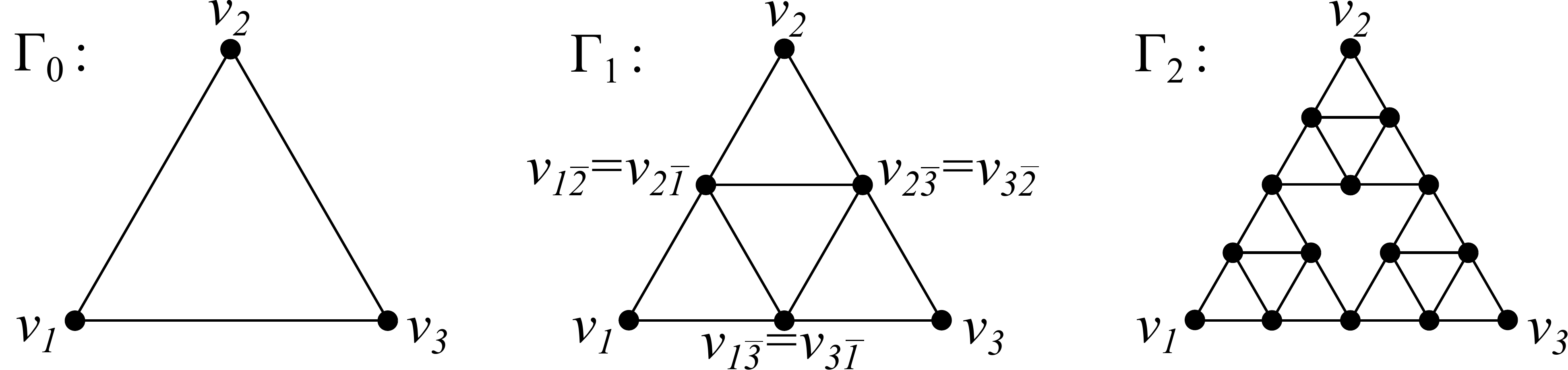}
	\caption{Graph approximations of $G$.}
	\label{fig:sg_graph}
      \end{figure}
  The continuum limit has proved to be an effective tool for studying coupled dynamical systems.
As in the related work on the KM on graphs \cite{Med14a, Med2026}, we derive the continuum limit
of \eqref{KM} as $n \to \infty$. The scaling factor $\left(\frac{5}{3}\right)^n$ in front of the
sum on the right-hand side of \eqref{KM} ensures that the system admits a nontrivial continuum
limit as $n \to \infty$.

  In the limit of large $n$, we expect that \eqref{KM} becomes the heat equation on the SG
  \begin{equation}\label{heat}
    \partial_t u(t,x) =\Delta u(t,x), \quad x\in G.
  \end{equation}
  The expected form of the continuum limit \eqref{heat} is supported both by the analysis of steady states
  of the KM on the SG in the follow-up work \cite{MedMiz2025} and by the study of the KM on
  random geometric graphs \cite{CirGro2025}. If this conjecture is confirmed, the KM on SG along with the
  model in \cite{CirGro2025} will provide an interesting  example of a \textit{nonlinear} spatially extended
  dynamical system, which admits a continuum limit in the form of a \textit{linear} PDE
  (cf.~\cite{MedMiz2025c}).

  Before turning to the continuum limit of the KM on the SG, we outline several preparatory steps.
  The technical difficulties already manifest themselves
  in the analysis of steady state solutions of \eqref{KM} on $\Gamma_n$ for large $n$.

  Setting the right-hand side of \eqref{KM} to zero and expanding $\sin$ to first order yields the following equation
  for steady state solutions:
\begin{equation}\label{Taylor}
\left(\tfrac{5}{3}\right)^n \sum_{j:~j\sim_n i}\big(u_j-u_i\big)
+\left(\tfrac{5}{3}\right)^n \sum_{j:~j\sim_n i} g(u_i,u_j)=0, \quad i\in V_n\setminus V_0,
\end{equation}
where $g$ denotes the nonlinear terms from the Taylor expansion of $\sin$. From the construction of the
Laplacian on the SG (cf.~\cite{Kig01}), one expects the first term on the left-hand side of \eqref{Taylor}
to converge to $\Delta u(x)$ as $n\to\infty$. Then one has to show that the second sum vanishes in the limit.

This program is carried out in \cite{MedMiz2025} via $\Gamma$-convergence methods. A fundamental obstacle,
however, becomes apparent already at this stage: both the theory of the Laplacian on the SG (cf.~\cite{Kig01}) and
$\Gamma$-convergence (cf.~\cite{Braides-beginners}) are formulated for real-valued functions, whereas
here the solutions take values in the torus $\T$. What might appear to be a minor technical difficulty turns
out to have significant consequences: the topology of the codomain strongly influences the structure of
solutions and dictates the analytical framework required for their analysis.

The analysis of both terms in \eqref{Taylor} requires care. In this work, we concentrate on
the first term,
corresponding to the linearized equation for $\T$-valued steady states on the SG:
\be\label{Lap}
\Delta u(x) = 0, \qquad x \in G \setminus V_0,
\ee
with suitable boundary conditions prescribed on $V_0=\{v_1,v_2,v_3\} \subset G$.
We refer to $\T$-valued solutions of
\eqref{Lap} as harmonic maps (HMs) on the SG.
\begin{figure}[h]
	\centering
 \textbf{a}\includegraphics[width = .4\textwidth]{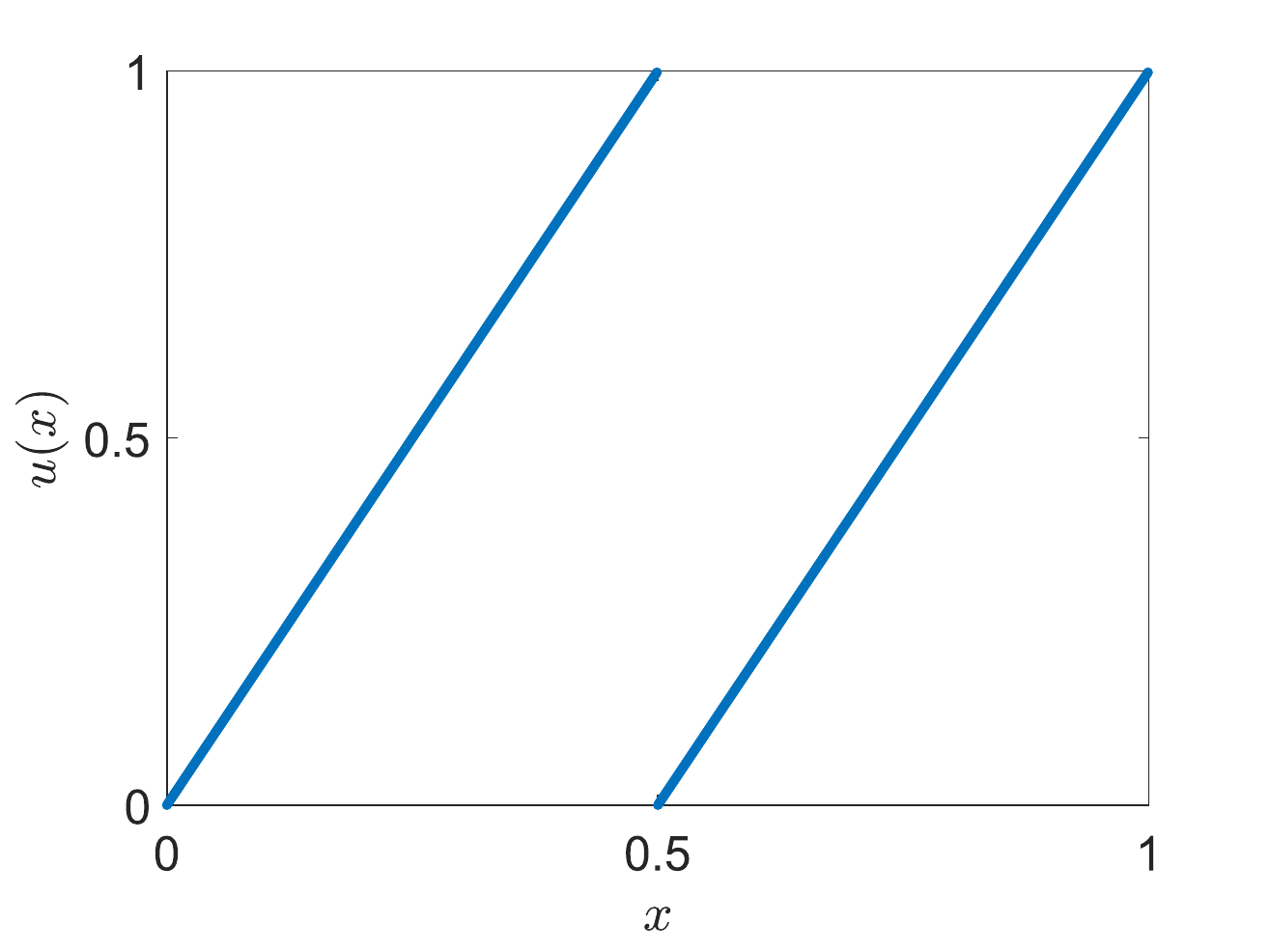}\quad
 \textbf{b}\includegraphics[width = .4\textwidth]{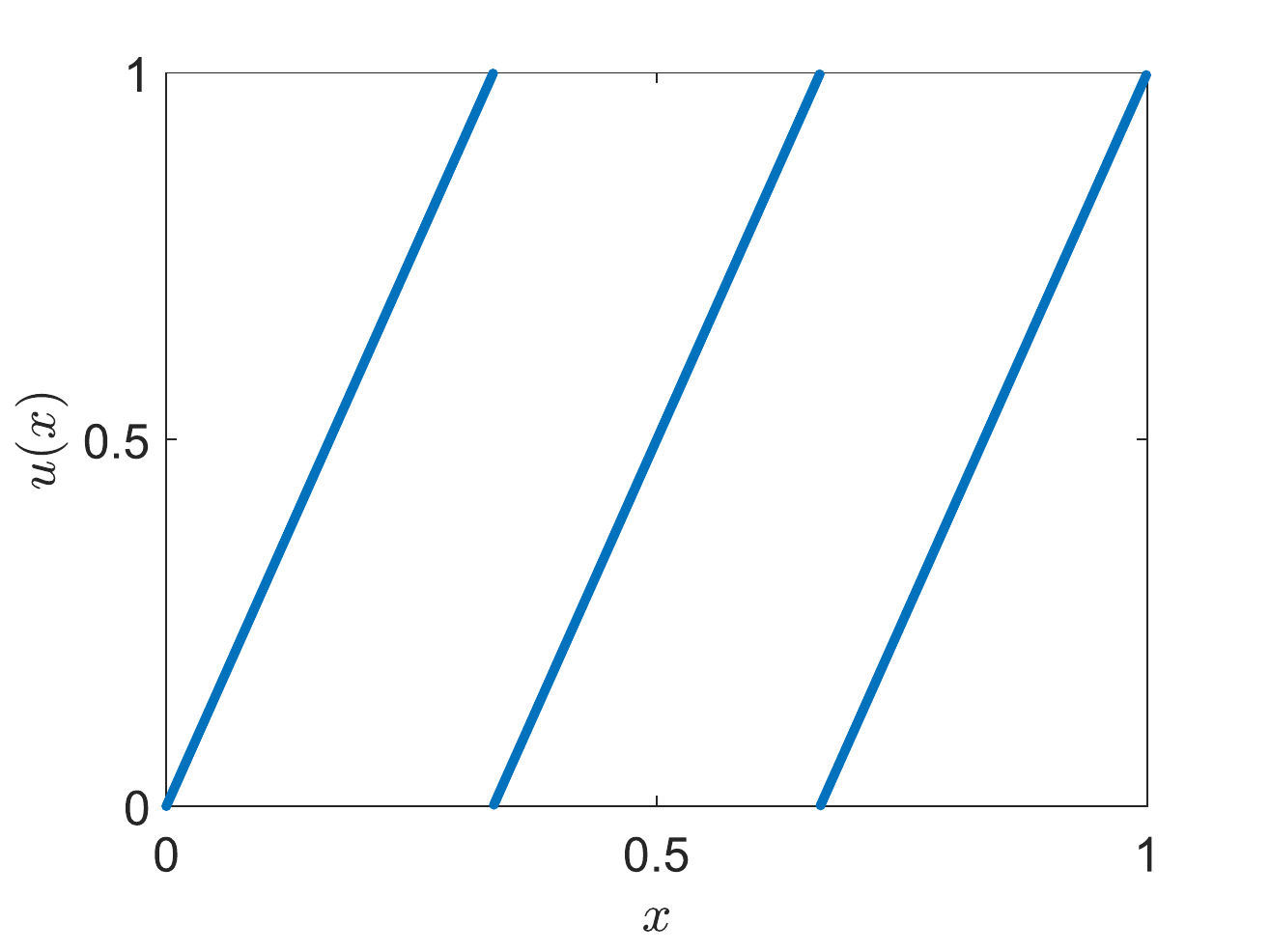}
	\caption{\textbf{a} $\T$-valued solutions of \eqref{Lap} on $\T$:
 \textbf{a}~$2$-twisted state, \textbf{b}~ $3$-twisted state.}
	\label{fig:euclidean}
\end{figure}

The following  examples illustrate how the topology of the codomain influences the structure of solutions
of \eqref{Lap}. Consider the boundary value problem for the Laplacian on $\T$:
\be\lbl{R-bvp}
\Delta u = 0, \qquad u(0)=u(1).
\ee
The only real-valued solutions of \eqref{R-bvp} are constants, $u\equiv c$ with $c\in\R$;
in other words, up to an additive constant, the unique solution is $u\equiv 0$. By contrast,
if $u$ takes values in $\T$, there exist
infinitely many topologically distinct stable solutions,
\be\lbl{q-twist}
u_q(x) = qx + c \pmod 1, \qquad q\in\Z,\; c\in\T.
\ee
Here $q$ denotes the degree of the continuous map $u:\T\to\T$. By the Hopf theorem,
the degree determines the homotopy class of solutions to \eqref{R-bvp}. Thus,
modulo translations, there is a unique HM from $\T$ to itself in each homotopy class.

 \begin{figure}[h]
	\centering
	\textbf{a}\includegraphics[width  =.4\textwidth]{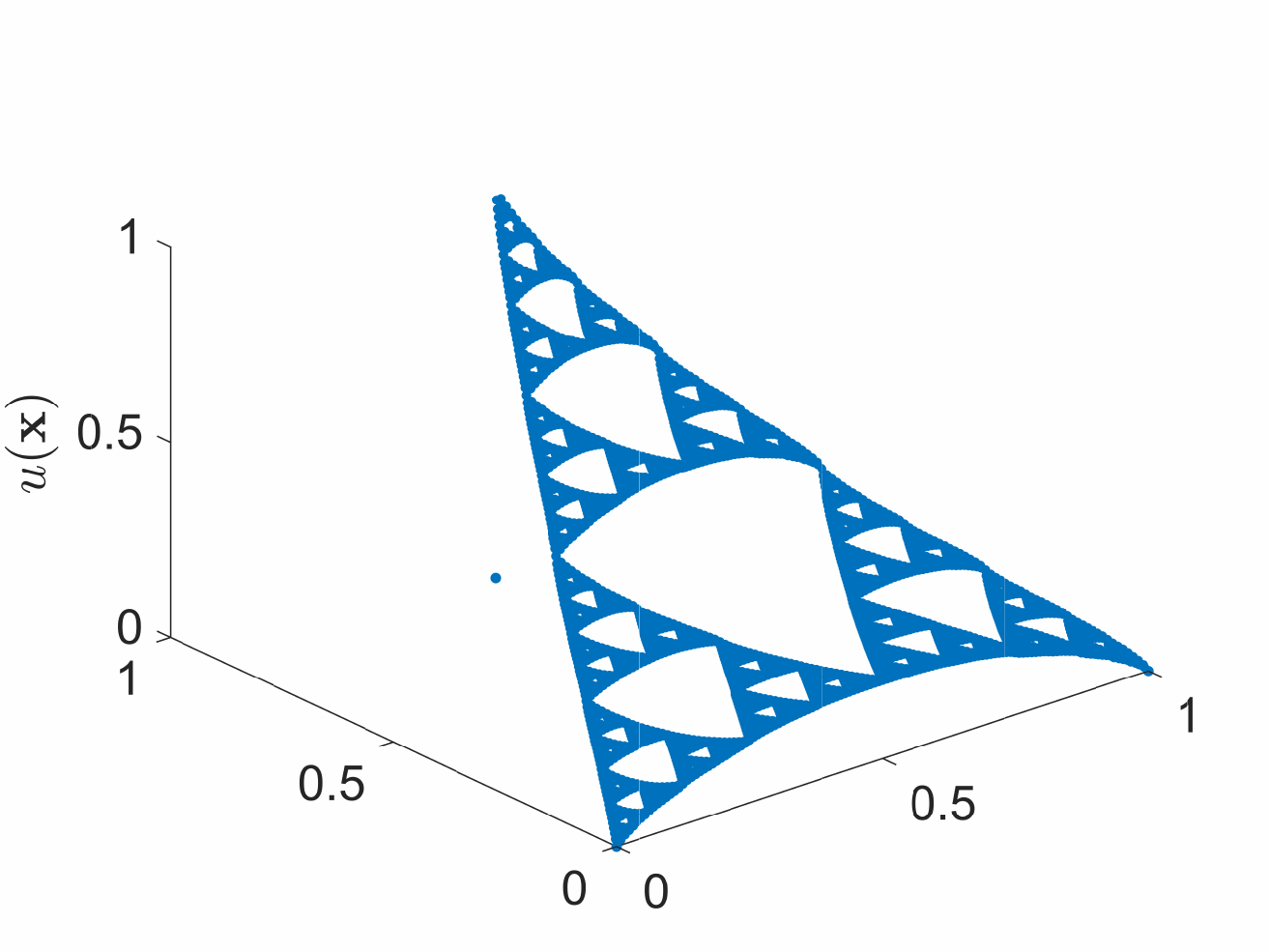}
 \textbf{b}\includegraphics[width = .4\textwidth]{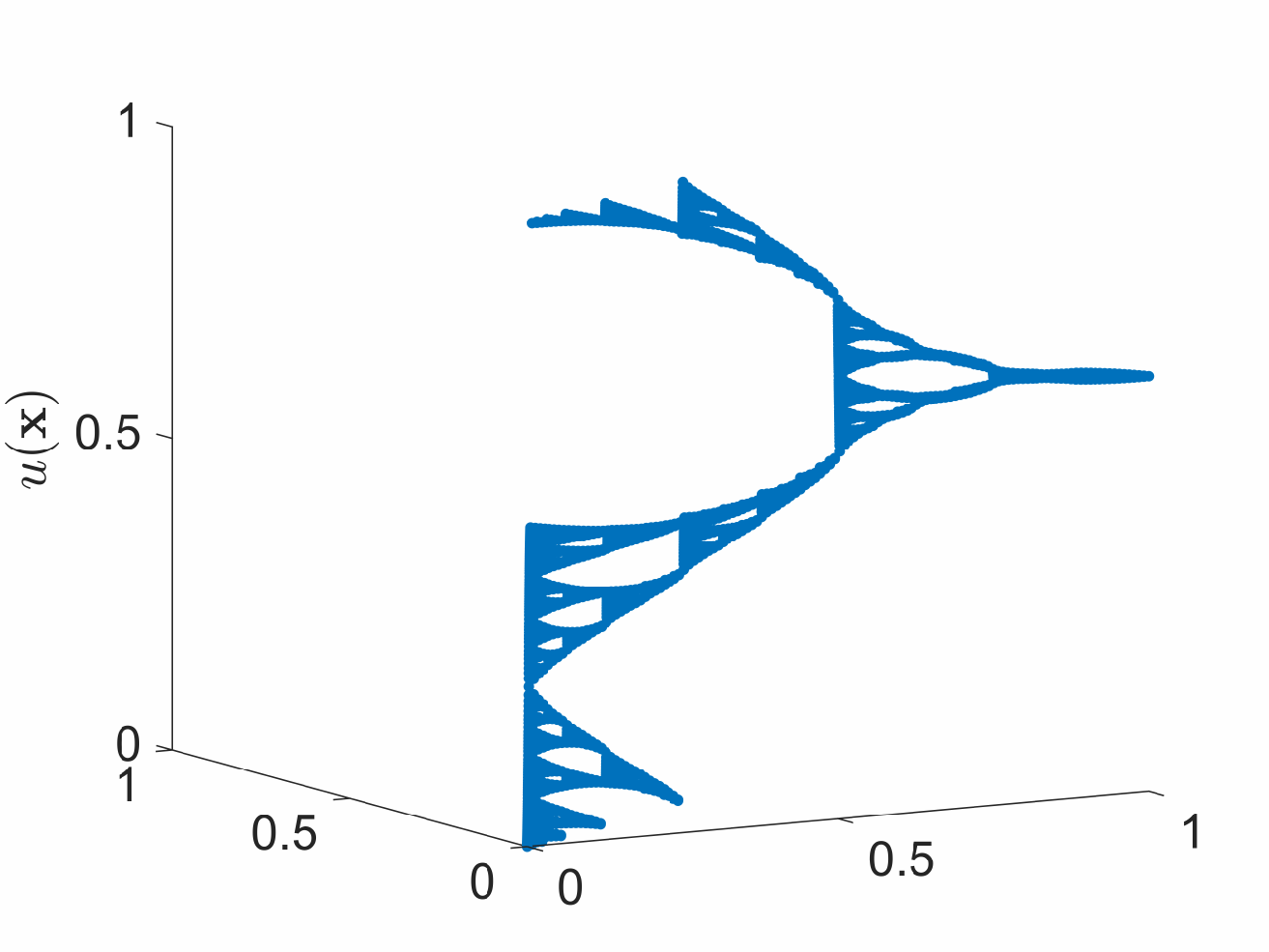}\quad
 \textbf{c}\includegraphics[width = .4\textwidth]{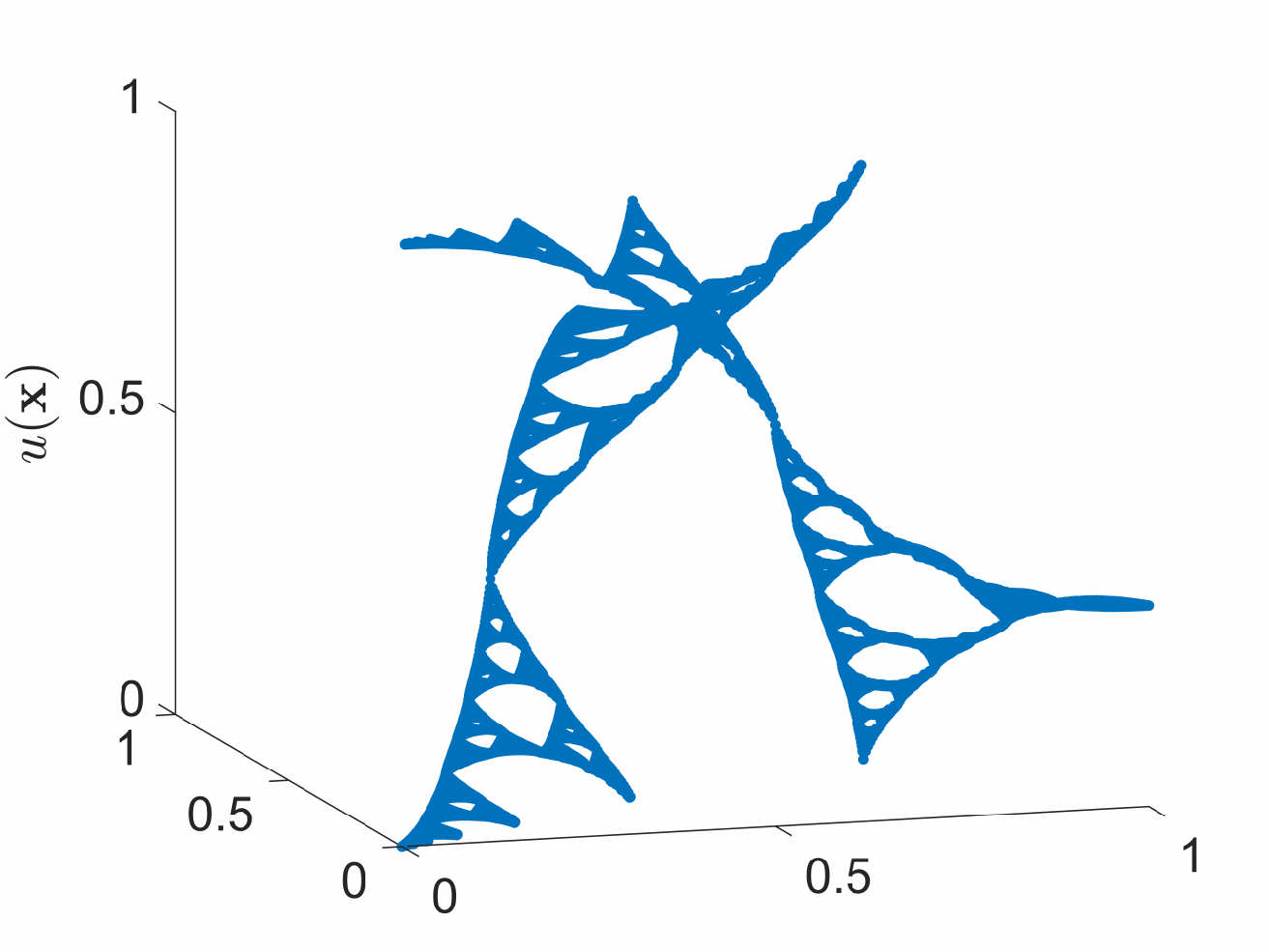}\quad
 \textbf{d}\includegraphics[width = .4\textwidth]{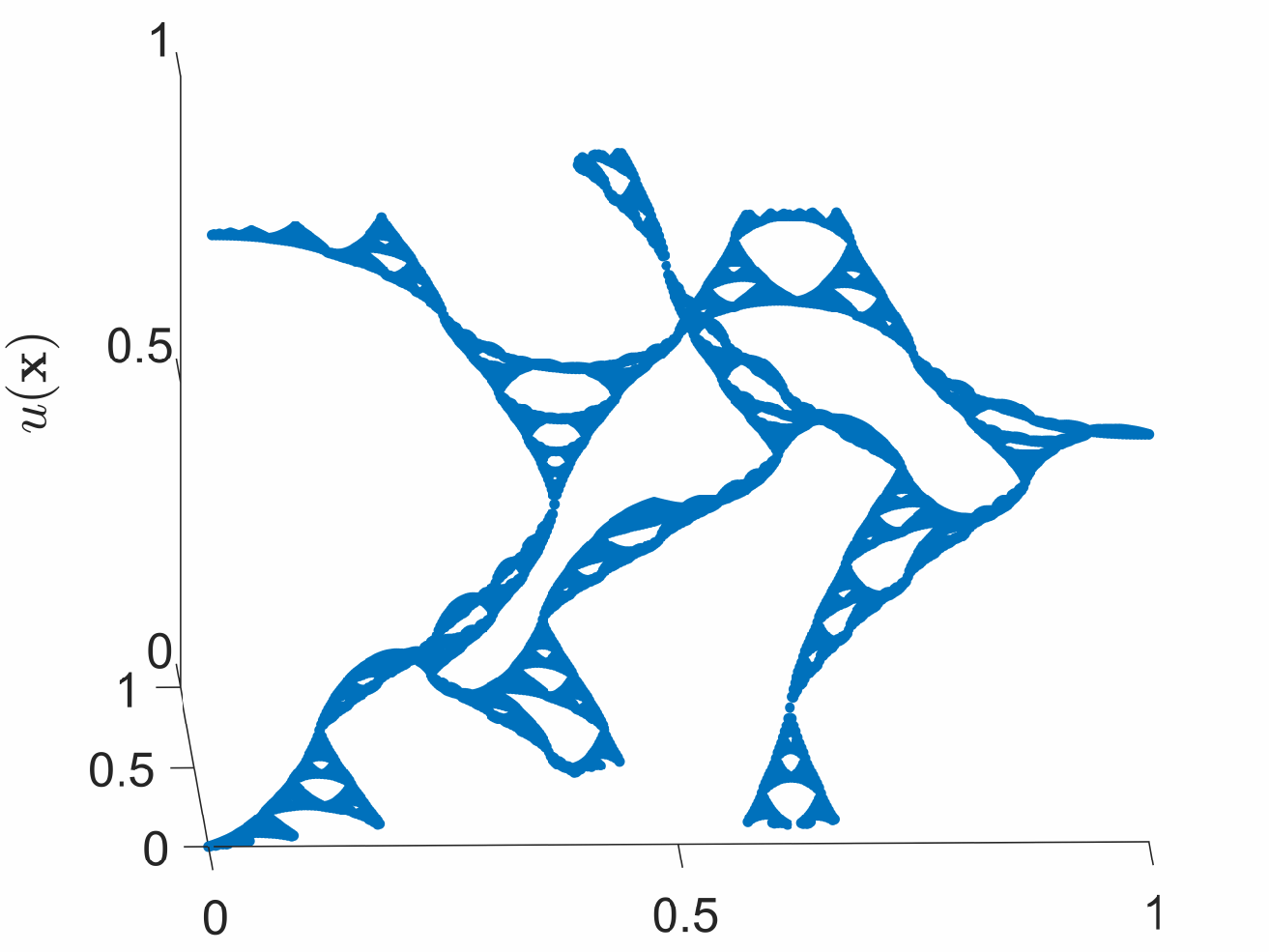}
	\caption{$\T$-valued solutions of \eqref{Lap} of different degrees:
	\textbf{a}~ $(0)$ \textbf{b}~ $(1)$, \textbf{c}~ $(2)$, \textbf{d}~ $(1 1 1 1)$. }
	\label{fig:sg1}
\end{figure}

As a fractal domain, $G$ exhibits a richer topology than the unit circle considered in the previous
example.
There are infinitely many independent loops in the
  SG, generated by the graphs approximating SG (see Fig.~\ref{fig:sg_graph}). For instance,
  the boundary of each triangular cell $T_w, w\in \mathcal{S}^n, n=0,1,2,\dots,$ yields a cycle. Denote the
  loops corresponding to $\partial T, \partial T_{(1)}, \partial T_{(2)},\dots$ by $\gamma_0, \gamma_1,
  \gamma_2,\dots,$ respectively. A HM $u$ restricted to each of these $\gamma_i, i=0,1,2,\dots$
  may have a nontrivial winding number $\omega_i$. These winding numbers impose global topological constraints on
  solutions of the boundary value problems for HMs on SG. Below, we introduce the degree vector
  $\bar\omega(u)=(\omega_0, \omega_1, \omega_2,\dots)$, which specifies the homotopy class for $u$.

The fact that the homotopy class of a map $u : G\rightarrow \T$ is
characterized by the degree vector $\bar\omega_\gamma(u)$ can be viewed as a consequence
of the fact that $\T$ is an Eilenberg-MacLane space of type $K(\Z, 1)$. This implies a natural isomorphism
$[K, T]\cong H^1 (G; \Z).$ From this perspective, $\bar\omega(u)$ is algebraically the coordinate
representation of a cohomology class in $H^1(G; \Z)$ with respect to a chosen basis of
$H^1(G; \Z)$. This explains why the topological constraints are additive and Abelian.

In view of this discussion, one expects a far more diverse family of HMs from $G$ to $\T$
than from $\T$ to itself. The HMs illustrated in Figure~\ref{fig:sg1} confirm this expectation:
each plot depicts HMs on the SG corresponding to different homotopy classes. In the follow-up
work \cite{MedMiz2025}, we show that each homotopy class gives rise to a topologically
distinct steady-state solution of \eqref{KM} in the limit $n\to\infty$.

The problem of HMs from the SG to the circle was first studied by Strichartz \cite{Strich02},
who used explicit algebraic calculations based on the harmonic extension algorithm from \cite{Kig01}
to show that a unique HM exists in each homotopy class under suitable boundary conditions.
While this approach provides an algorithmic construction of HMs, it is not suited for the analytical study
of $\T$-valued solutions of \eqref{Lap} and \eqref{Taylor}, which is our primary focus.

Here, we propose a geometric method for constructing HMs from the SG to the circle
  Our method is based on a covering space constructed separately for each degree vector, i.e., for each homotopy
  class. The construction of the covering space in the simplest nontrivial case $\bar\omega(u)=(1,0,0,\dots)$ when
  the winding number over the outer boundary of SG  is equal to $1$ and there are no other loops with nonzero
  winding numbers, is shown in Fig.~\ref{fig:sg2}. In a nutshell,  the idea is to take the infinite product
  $G_\times=G\times \Z$,
  then make cuts in the domains at each level and identify the cut points from two consecutive levels,
  as shown in Figure~\ref{fig:f-cuts}, to form a connected
  covering space $\tilde G$.
  We can now lift the $\T$-valued solution $u$ on $G$ to a real valued solution $\tilde u$ on $\tilde G$ preserving the
  topological constraint given by the winding number.

  After that we restrict to the fundamental domain $G^0=G\times \{0\}$. The topological constraint on the $\T$-valued
  function $u$ now translates into the jump boundary condition for $\tilde u$ at the cut points $z^0_-$ and $z^0_+$
  (see Fig.~\ref{fig:sg2}). The problem is reduced to finding a real-valued harmonic  function $\tilde u$
  on the fundamental domain $G^0$ subject to the corresponding boundary conditions. To solve this problem we adapt
  a well-known \textit{harmonic extension algorithm} (cf.~\cite{Str06}). On self-similar domains like SG, harmonic
  extension provides an efficient algorithm for recursive computation of harmonic functions restricted to
  the set of vertices of $\Gamma_m$
  by starting from the boundary conditions given at the nodes of $\Gamma_0$.

  The combination of the covering space, which takes care of the global topological constraints, and the harmonic
  extension, which effectively solves the harmonic problem locally, provides a powerful technique for computing
  $\T$-valued HMs on fractal domains, which extends naturally to other p.c.f. fractals, far beyond the SG
  considered in \cite{Strich02}. Importantly, the lifting procedure translates the problem for $\T$-valued functions
  into the real-valued
setting, enabling the application of techniques of classical analysis. This forms the basis for extending
$\Gamma$-convergence methods to $\T$-valued solutions in \cite{MedMiz2025}. We also formalize the notion
of homotopy classes in this setting and establish an analogue of the Hopf degree theorem for $\T$-valued
continuous functions on the SG, which was not addressed in \cite{Strich02}.

Our method relies on two key topological properties of the fractal domain at hand. 
The first is compactness, which implies \textit{uniform continuity} of solutions. The latter
property serves as the structural linchpin of our approach. 
It allows the problem to be decomposed into two distinct regimes:

\begin{itemize}
\item \textbf{Global topology:} 
At large scales, the nontrivial topology of $\mathbb{T}$ manifests itself through winding numbers, which are resolved via the construction of the covering space. 
Uniform continuity plays an essential role in the characterization of homotopy classes of homomorphisms: it permits one to restrict functions to the dense subset 
\[
\bigcup_{m=0}^{\infty} \Gamma_m \subset G,
\]
and then extend them to $G \setminus \bigcup_{m=0}^{\infty} \Gamma_m$ by continuity.

\item \textbf{Local linearization:} 
  At the micro scale, uniform continuity guarantees that the variation of the function
  is small enough to allow the solution to be treated locally as real-valued.
\end{itemize}

Another ingredient, whose importance may not be immediately apparent, is that the fractals
under consideration are \textit{finitely ramified}. Roughly speaking, this means that a connected fractal
can be disconnected by removing only finitely many points. This structural property enables us to
introduce cuts that resolve global topological constraints along finitely many cycles (namely, the homology generators),
while leaving the homological properties of the remaining domain unaffected. In this way, we are able
to perform cuts with surgical precision.

Our main result, Theorem~\ref{thm.main}, states that \eqref{Lap} on the SG, with appropriate boundary conditions,
admits a unique solution in each homotopy class. Consequently, as with HMs from $\T$ to itself,
there exist infinitely many $\T$-valued solutions of the Laplace equation on the SG. The hierarchy of
HMs on $G$ is captured by the integer-valued vector, the degree intrinsic to each homotopy class (cf.~\eqref{o-vector}).
Analogously to continuous maps between compact Riemannian manifolds (cf.~\cite{EelSam64}), HMs
from the SG to the circle represent continuous functions in each homotopy class.

Theorem~\ref{thm.main}, the main result of this work, parallels the classical Hodge Theorem (e.g., on manifolds with boundary),
where harmonic objects are defined as minimizers of a Dirichlet energy functional.
In the classical setting, one typically seeks a unique harmonic $1-$form that minimizes the
$L_2$-energy subject to the real cohomology $H^1(G; \R)$ \cite{Jost-Riemann}. In our case, the problem
is to find a map
minimizing energy subject to a prescribed degree vector, which requires the restriction to
integer coefficients $H^1(G; \Z)$ and naturally selects a discrete lattice within the space of
harmonic forms, preserving the uniqueness guaranteed by the convexity of the energy.
\footnote{The authors are grateful to an anonymous reviewer whose insightful 
  observations, including the Eilenberg–MacLane and Hodge perspectives, helped
  to better contextualize the results of this work.}

The paper is organized as follows. In Section~\ref{sec.maps}, we discuss continuous maps from the SG to the circle,
defining the degree of a HM and showing that it determines the homotopy class, as in the case
of maps from a circle to itself. Section~\ref{sec.harmSG} reviews relevant results on real-valued harmonic functions
on the SG. Sections~\ref{sec.cover} and \ref{sec.harm-struct} describe the construction of the covering
space and the associated harmonic structure.
 Section~\ref{sec.HM} presents the construction of HMs from the SG to the circle. For clarity, we
 first consider the simple case where the degree is a single integer, $\bar{\omega}(u)=(\rho_0)$, and then
 extend the construction to HMs of arbitrary degree in Section~\ref{sec.higherorderSG}. Finally,
 Section~\ref{sec.pcf} discusses the generalization of our method to p.c.f. fractals, for which the
 harmonic extension algorithm is also available. We highlight the features of the SG used in the proof of
 Theorem~\ref{thm.main} that do
 not generalize directly and propose strategies to address these challenges.

\section{The main result}\label{sec.maps}
\setcounter{equation}{0}

We begin with the description of 
the structure of the loop space of $G$. Throughout this paper, $G$ denotes SG.
Recall that $T$ stands for a solid closed triangle with vertices $v_1, v_2,$ and  $v_3$
and let $S=\{1,2,3\}$ stand for the alphabet of three symbols. For $n\in\N$
and $w=(w_1,w_2,\cdots w_n)\in S^n$, we define an $n$-cell as
$$
T_w=F_w(T)\doteq F_{w_1}\circ F_{w_2}\circ \dots \circ F_{w_n} (T)
$$
and $\partial T_w$ stands for an oriented boundary of $T_w$. In addition,
$T_\emptyset=T$.

Let $\gamma_0, \gamma_1, \gamma_2, \dots$
stand for the loops from $\mathcal{P}=\left\{\partial T_w, w\in\bigcup_{n=0}^\infty S^n\right\},$
with $w$
  taken in  lexicographical order, i.e.,
  $\gamma_0=\partial T_\emptyset, \gamma_1=\partial T_{1}, \gamma_2=\partial T_{2}, $
  etc.

  Given $f\in C(G,\T)$, the restriction of $f$ to $\gamma\in\mathcal{P}$ after appropriate
  reparametrization $c_\gamma:\gamma\to\T$ yields a map from $\T$ to itself,
  $f_\gamma\doteq f\circ c_\gamma$. 
  Let $\omega(f_\gamma)$ denote the degree of $f_\gamma\in C(\T,\T)$.
  The degree of $f\in C(G,\T)$ is defined as follows
\begin{equation}\label{o-vector}
  \bar\omega (f)=\left( \omega_{\gamma_0} (f), \omega_{\gamma_1} (f), \omega_{\gamma_2} (f), \dots  \right),
  \end{equation}
  where $\omega_{\gamma}(f)\doteq \omega(f_\gamma)$ is the degree  of $f$ restricted to
  a closed simple path $\gamma\in\mathcal{P}$ and reparametrized appropriately.
Below, we show that \eqref{o-vector} determines the homotopy type
of $f$ as a map from $G$ to $\T$. 

Since $f\in C(G,\T)$ is uniformly continuous and the diameter of $T_w$ goes
  to $0$ as the length of $w$ goes to infinity,
$\bar{\omega}(f)$ can have only finitely many nonzero entries.
 Consequently, there exists 
  $N\in\N$ such that the winding numbers for 
  loops  bounding $m$-cells for $m> N$  are all zero:
  $$
  \omega_{\gamma_i}(f)=0, \; i> \frac{3^{N+1}-3}{2}.
  $$
  Therefore, the homotopy class of a HM on $G$ is determined by the degree
  over loops bounding $m$-cells for $m\le N$, i.e., by a finite number of entries
  in \eqref{o-vector}
\begin{equation}\label{degree}
 \bar{\omega}(f)=\left( \omega_{\gamma_0}(f), \omega_{\gamma_1}(f),
    \omega_{\gamma_2}(f), \dots, \omega_{\gamma_{ \frac{3^{N+1}-3}{2} }}(f)\right),
\end{equation}
where we dropped the infinite sequence of zeros at the end.



Following \cite{Strich02}, HMs on SG and other p.c.f. fractals will be constructed as solutions of a certain boundary value
problem.
To set up the problem for $G$,
we need to understand the homotopy classes
of continuous functions on $G$ first. To this end, let
        $$
        \mathcal{P}_n=\left\{ \partial T_w,\; w\in S^n\right\} \quad \mbox{and}\quad
        \mathcal{P}= \bigcup_{n=0}^\infty   \mathcal{P}_n.
        $$
\begin{figure}[h]
		\centering
		\includegraphics[width = .7\textwidth]{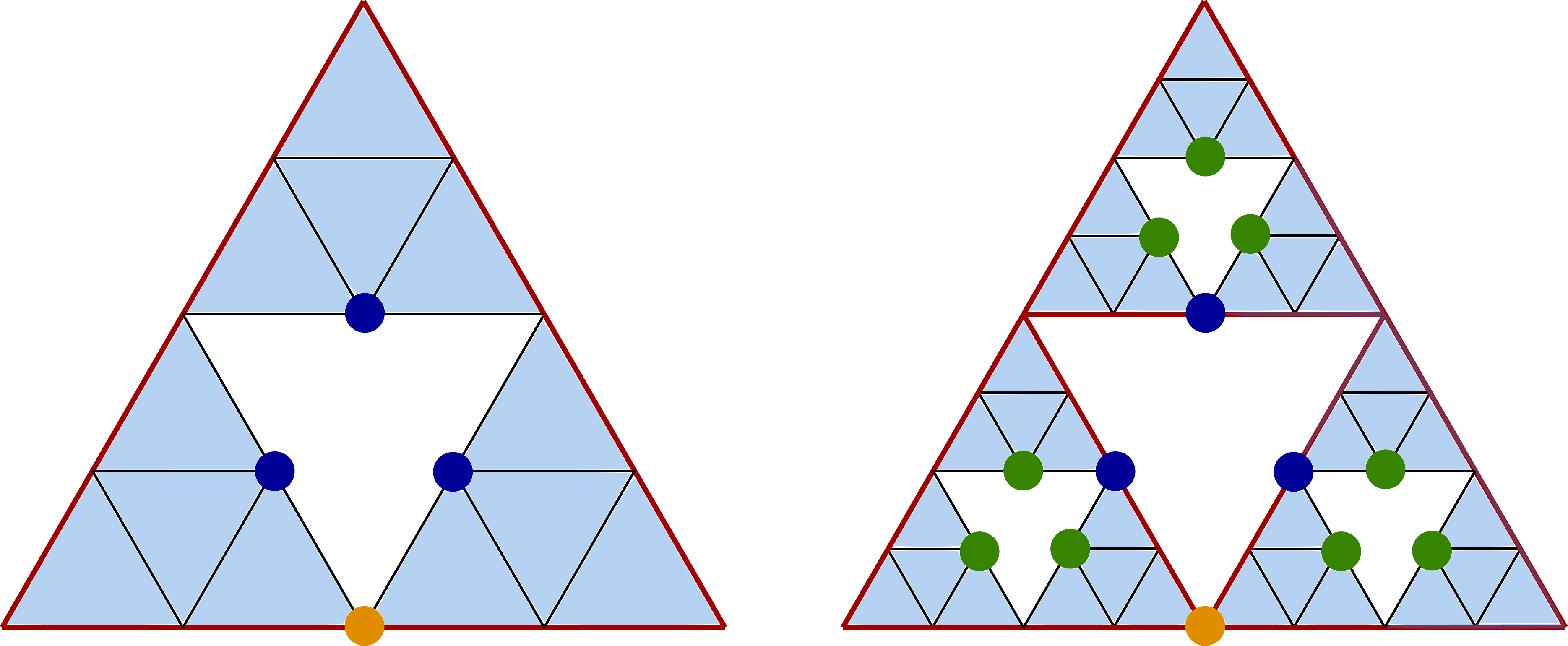}
		\caption{(left)  Reference point $\xi^0$ for the 
                 outer loop $\gamma_0 = \partial T_\emptyset$ is shown in red
                 $\xi^i, \; i\in\{1,2,3\}$ are shown in blue.
                 (right) The next set  of reference points:
                 the reference point $\xi^k, \; k\in \{ 4,5,\dots, 12\}$ are shown in green.
                 They are chosen in the middle of the sides of the triangular of this level that do
                 not intersect 
                 $\gamma_0,\gamma_1,\gamma_2,$ or $\gamma_3$.}
		\label{fig:f-cuts}
	\end{figure}

  Next we want to parametrize cycles $\gamma_0, \gamma_1, \dots$. To this end, for each
  $\gamma_i,\; i=0,1,2,\dots,$
        we choose a reference point $\xi^i$ as follows: $\xi^0$ is chosen in the middle of the base of
        the triangular loop $\gamma_0$. $\xi^i, \; i\in \{1,2,3\}$ is chosen in the middle of the side of
        $\gamma_i\in \cP_1\setminus\cP_0$ (see Figure~\ref{fig:f-cuts}).
        We continue by induction, for $\gamma\in\cP_n,\;n\ge 2$,  the reference point $\xi\in\gamma$
        is chosen in the middle of the side of $\gamma$ that does not belong to $\cP_{n-1}$ (see Figure~\ref{fig:f-cuts}).

For $\gamma_k\in \mathcal{P}$, we choose a \textit{uniform} parametrization
$c_{\gamma_k}:\T\to\gamma$, which starts at $\xi^k$, $c_{\gamma_k}(0)=\xi^k$,
and traces $\gamma_k$ in clockwise
direction with
constant speed $1/|\gamma_k|$. Here, $|\gamma|$ stands for the length of $\gamma$.

For a given $f:~G\rightarrow \T$,  $f_\gamma\doteq f\circ c_\gamma$ is a function
from $\T$ to itself.

There is   a unique continuous  lift $\bar{f}_\gamma:\R\to\R$
such that
\begin{align} \label{lift-1}
  \bar{f}_\gamma(0)&= f(0),\\
  \label{lift-2}
  \bar{f}_\gamma(x)&= f(x\mod 1) + k(x), \quad k(x)\in \Z.
\end{align}
The second condition \eqref{lift-2} can be written as 
\begin{equation}\label{lift-3}
\pi\circ \bar f_\gamma=\bar f_\gamma \circ \pi,
\end{equation}
where $\pi: \,\R\ni x\mapsto x\mod 1.$ 

The degree of $f_\gamma$ is expressed in terms of the lift of 
$f_\gamma$:
$$
\omega(f_\gamma)\doteq \bar{f}_\gamma(1)-\bar{f}_\gamma(0).
$$

\begin{definition}\label{def.w-vector}
  The degree of $f\in C(G,\T)$ is defined by
  \begin{equation}\label{w-vector}
    \bar\omega(f)=\left(\omega_{\gamma_0}(f), \omega_{\gamma_1}(f), 
    \omega_{\gamma_2}(f),\dots \right).
\end{equation}
\end{definition}

\begin{remark}
    Since $f$ is uniformly continuous on $G$, the number of nonzero entries in \eqref{w-vector}
    is finite.
\end{remark}

\begin{definition}\label{def.homotopy}
  Two maps $f,g \in C(G,\T)$ are called homotopic, denoted $f\sim g$, if there exists a continuous mapping
  $F: [0,1]\times G\to \T$ such that
  \begin{equation}\label{F-hom}
    F(0,\cdot)=f\qquad\mbox{and}\quad F(1,\cdot)=g.
    \end{equation}
  \end{definition}

  \begin{theorem}\label{thm.homotopy} Let $f,g \in C(G,\T)$. Then
    $f\sim g$ if and only if $\bar\omega(f)=\bar\omega(g)$.
  \end{theorem}
  \begin{proof} See~Appendix~A.
\end{proof}

We now turn to the Laplace equation
\begin{equation}\label{Laplace}
\Delta f(x)=0, \qquad x\in G\setminus V_0,
\end{equation}
where $f\in C(G,\T)$.

In analogy to the formulation of the Dirichlet problem for real-valued  
functions on $G$ \cite{Kig01}, 
we include the values of $f$ on $V_0$  in the formulation
of the boundary value problem for HMs. In fact, we need a little more. Note that eventually 
we will need to recover $f_{\gamma_0}$. Thus, instead of the values
of $f|_{V_0}$ we incorporate the values of $\bar f_{\gamma_0}$, the lift of $f_{\gamma_0}$,
to our formulation of the boundary value problem for $f: G\to \T$. Specifically, we 
use  
\begin{align}
\nonumber
\partial_{v_1v_2} f &\doteq \bar f_{\gamma_0} (1/3)-\bar f_{\gamma_0} (0)= \delta_{1},\\
\label{bc-hm}
\partial_{v_2v_3} f &\doteq \bar f_{\gamma_0} (2/3)- \bar f_{\gamma_0} (1/3)= \delta_{2},\\
\nonumber
\partial_{v_3v_1} f &\doteq \bar f_{\gamma_0} (1)-\bar f_{\gamma_0} (2/3)= \rho_0- \delta_{1} - \delta_{2}.
\end{align}
Here, $\rho_0\doteq \omega_{\gamma_0}(f)$.
\begin{remark}
	\label{rem:bc}
	To illustrate the importance of defining the boundary conditions as real-valued lifts, fix $\rho_0=0$ and first consider the trivial HM, $f_1\equiv 0$. 
 Next, let $f_2$ be the HM obtained by taking $\delta_{1} = 1 = -\delta_{2}$ in \eqref{bc-hm}; the result is shown in Figure \ref{fig:sg1}{\bf a}. As $\T$-valued maps, we see that $f_1$ and $f_2$ agree on $V_0$ and have the same degree, but $f_1\neq f_2$.
\end{remark}

In addition to \eqref{bc-hm}, in analogy to the case of HMs from $\T$ to $\T$,
we will need to specify the homotopy class of $f$:
\begin{equation}\label{homotopy-c}
\bar\omega(f)=\left(\rho_0, \rho_1, \dots, \rho_n,0,0, \dots\right) \in \Z^\ast\doteq \bigcup_{l=1}^\infty \Z^l.
\end{equation}

  For brevity, we drop the trailing zeros and write
$\bar\omega(f)=(\rho_0,\rho_1,\dots,\rho_n)$ instead of the longer expression in \eqref{homotopy-c},
assuming that $\rho_n$ is the last nonzero entry. If $\bar\omega(f)=(0,0,\dots)$, we write
$\bar\omega(f)=(0)$.

We can now state our main result for the SG.

\begin{theorem}\label{thm.main}
There is a unique HM $f: G\to\T$ satisfying \eqref{bc-hm} and \eqref{homotopy-c}.
\end{theorem}

\section{Harmonic structure on SG}\label{sec.harmSG}
\setcounter{equation}{0}
Before discussing $\T$-valued HMs on $G$, it is instructive to
review the definition and basic properties of real-valued harmonic
functions on  $G$ (cf.~\cite{Kig01}).

Recall the definition of $\Gamma_m, m\in \N$,  graphs approximating $G$ (see Fig.~\ref{fig:sg_graph}).
We will distinguish the
\textit{boundary} and \textit{interior} nodes of $\Gamma_m$. The former
are given by the vertices of $\Gamma_0$, i.e., $V_0=\{v_1, v_2, v_3\}$, and the latter are
the remaining nodes $V_m\setminus V_0$.  Let $L(X,Y)$ stand for the space of functions from $X$ to $Y$.

\begin{definition}
  $f \in L(V_m,\R)$ is called $\Gamma_m$-harmonic if it satisfies the discrete
	Laplace equation at every interior node of $\Gamma_m$:
	\begin{equation}\label{d-harm}
	\Delta_mf\,(x)\doteq\sum_{x\sim_m y} \left(f(y) -f(x)\right)=0\quad \forall y\in V_m\setminus V_0.
              \end{equation}
              \end{definition}

             Harmonic functions on the SG and other fractals are defined with the help of the 
             Dirichlet (energy) form $\mathcal{E}(f)$, which in turn is defined using properly 
             scaled energy forms on graphs approximating corresponding fractals (cf.~\cite{Kig01}).
The Dirichlet forms on the sequence of graphs $(\Gamma_m)$ are defined as follows:
	\begin{equation}\label{e-form}
          \mathcal{E}_{m}(f)=\left(\frac{5}{3}\right)^m
          \sum_{xy\in E(\Gamma_m)} \left(f(x)-f(y)\right)^2,\quad f\in L(V_m,\R),
	\end{equation}
where $E(\Gamma_m)$ denotes the set of edges of $\Gamma_m$.        

The sequence of Dirichlet forms $(\mathcal{E}_m)$ has the following 
properties.
\begin{enumerate}
    \item A $\Gamma_m$-harmonic $f_m^\ast\in L(V_m,\R)$ minimizes $\mathcal{E}_m$ 
    over all functions subject
    to the same boundary conditions
    \begin{equation}\label{variational-harmonic}
    \mathcal{E}_m(f_m^\ast)=\min\left\{ \mathcal{E}_m(f):\; f\in L(V_m,\R), \; f|_{V_0}=f_m^\ast|_{V_0}\right\}.
    \end{equation}
\item The minimum of the energy form over all extensions of 
$f\in L(V_{m-1},\R)$ to $V_m$ is equal to $\mathcal{E}_{m-1}(f)$:
\begin{equation}\label{variational-extension}
 \min\left\{ \mathcal{E}_m(\tilde f):\; \tilde f\in L(V_m,\R), 
 \tilde f|_{V_{m-1}}=f\in L(V_m,\R)\right\} = \mathcal{E}_{m-1}(f).
\end{equation}
\end{enumerate}

The first property follows from the Euler-Lagrange equation for $\mathcal{E}$.
This property does not depend on the scaling
coefficient $(5/3)^m$ in \eqref{e-form}. The second property, on the other hand, holds 
due to the choice of the scaling constant $(5/3)^m$.
The sequence of 
$(\mathcal{E}_m)$ is said to equip $G$ with a \textit{harmonic structure}.

By \eqref{variational-extension}, for  every $f\in C(G,\R)$
$\left( \mathcal{E}_m(f|_{V_m})\right)$ forms a nondecreasing
sequence. Thus,
$$
\cE(f)=\lim_{m\to\infty} \cE_m(f|_{V_m})
$$
is well-defined. The domain of the Laplacian on $G$ is defined by
$$
\operatorname{dom}(\cE)=\left\{ f\in C(G,\R):\; \cE(f)<\infty \right\}.
$$

\begin{definition}\label{def.R-harm}
A function $f\in \operatorname{dom}(\cE)$ is called harmonic if it 
minimizes $\cE(f)$ over all continuous functions on $G$ subject to given boundary conditions 
on $V_0$.
\end{definition}

      Property \eqref{variational-extension} yields
      a recursive algorithm for computing the values of a harmonic function on
      the union of sets of vertices of $(\Gamma_m)$,
      $V_\ast=\bigcup_{m=0}^\infty V_m$, a dense subset
      of $G$. For $m=0$, the values on $V_0$ are prescribed:
      $$
      f|_{V_0}=\phi.
      $$
      Given $f|_{V_{m-1}},\; m\ge 1,$ the values on $V_m/V_{m-1}$ are computed using the
      following $\frac{1}{5}-\frac{2}{5}$ rule, which we state for an arbitrary fixed $(m-1)$-cell
      $T_w,\; w\in S^{m-1}$: Suppose the values of $f$ at $a, b, c$, the nodes of $T_w$, are
 known. Then the values of $f$ at $x,y,z,$ the nodes at the next level of discretization  are 
 computed as follows 
\begin{equation}\label{classical-extension}
\begin{pmatrix} f_x \\ f_y\\ f_x 
\end{pmatrix}
=
\begin{pmatrix}
    \frac{2}{5} &  \frac{2}{5} &  \frac{1}{5}\\
     \frac{1}{5} &  \frac{2}{5} &  \frac{2}{5}\\
      \frac{2}{5} &  \frac{1}{5} &  \frac{2}{5}
\end{pmatrix}
\begin{pmatrix} f_a \\ f_b\\ f_c
  \end{pmatrix},
\end{equation}
where $f_v$ stands for the value of $f$ at $v\in V_m$
(Figure~\ref{fig:harm_ext}).
\begin{figure}[h]
	\centering
	\includegraphics[width  =.3\textwidth]{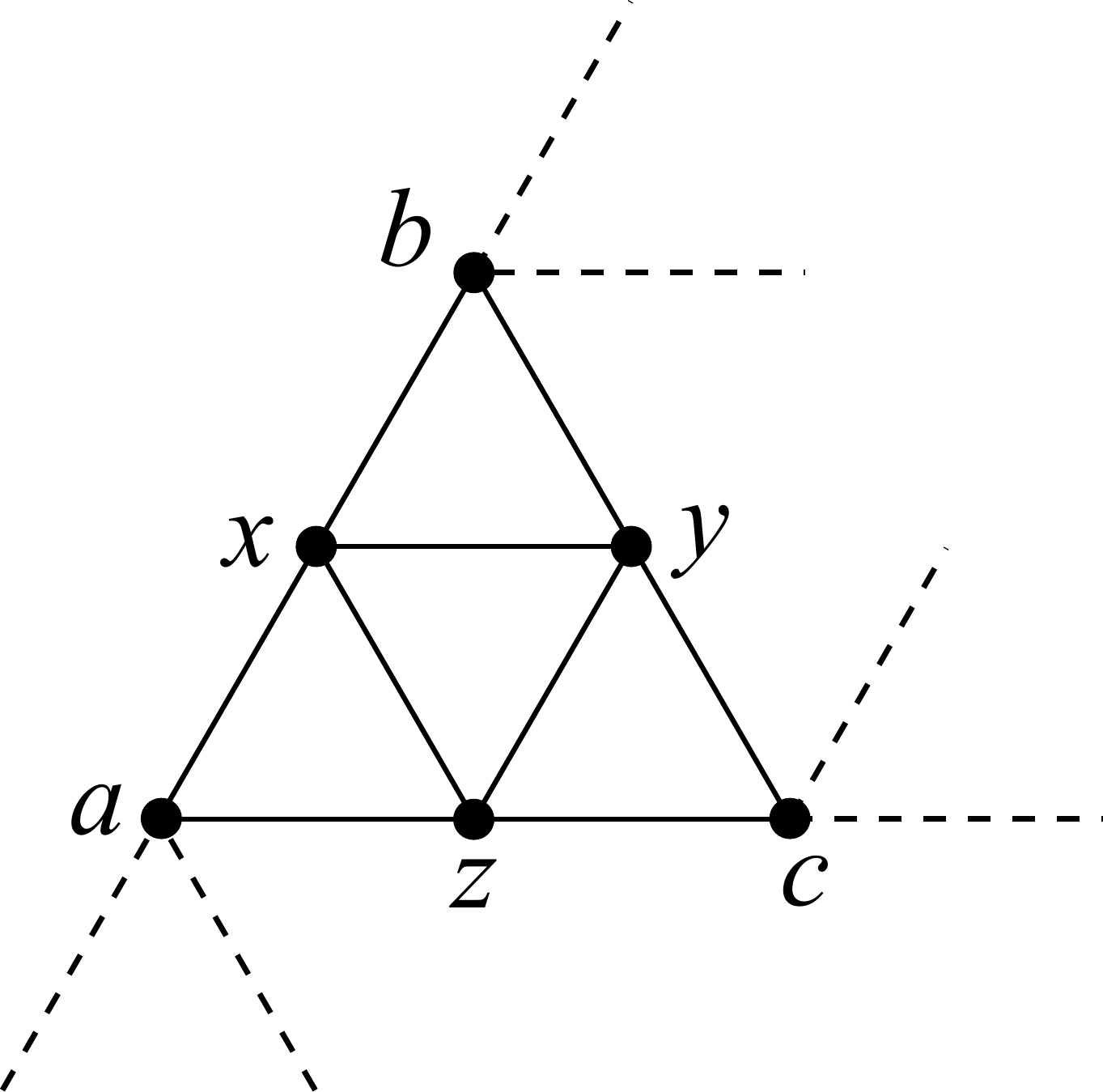}
	\caption{Harmonic extension algorithm for harmonic functions, see \eqref{classical-extension}.}
	\label{fig:harm_ext}
\end{figure}

      The  recursive algorithm of computing the values of a harmonic function $f$ on $V_{m+1}$
      using its values on $V_m$ via 
      \eqref{classical-extension} 
      is called \textit{harmonic extension}. 
      
      Using \eqref{variational-extension}, one can show that at each step of the harmonic extension:
      \begin{equation}\label{min-E}
      \cE_m(f|_{V_m})=\cE_0(f|_{V_0})=\min \left\{ f\in L(V_m,\R):\; f|_{V_0}=\phi\right\}.
      \end{equation}
      Thus, each of  $f|_{V_m}$ is a $\Gamma_m$-harmonic function. Further, harmonic extension
      results in a uniformly continuous function on $V_*$. Thus, the values of $f$ on $G\setminus V_*$ can
      be obtained from $V_*$ to SG by continuity (see \cite[\S 11.2]{Saenz-HarmAnal} for more details). 
      The continuous function obtained in this way minimizes $\cE$ over all continuous functions 
      on SG with the same boundary conditions. This yields a harmonic function $f$ on SG.

      By construction, $f\in \operatorname{dom}(\cE)$ and $f|_{V_m}$ is also $\Gamma_m$-harmonic
      for every $m\in \N$. This property may be used as an alternative definition of a harmonic function
      on SG. We refer the interested reader to \cite{Saenz-HarmAnal, Str06}
      for more details on harmonic extension.
      
%
%


  

\section{The covering space}\label{sec.cover}
\setcounter{equation}{0}

In this section, we explain the construction of covering spaces for $G$, which are the key
tool in the proof of Theorem~\ref{thm.main}. The covering space  is constructed separately for each 
degree vector
$\bar\omega(f)$.   We begin with the simplest nontrivial case of 
$
\bar\omega (f)=(\rho_0),
$
where we assume that $\rho_0\neq 0$, because otherwise
$f$ can be computed using the standard harmonic extension algorithm for real valued
functions (cf.~\eqref{classical-extension}).

Before we turn to the construction of the covering space, we need
 to explain the  coding of the nodes of the graphs
approximating $G$. Every node from $V_n$ is a vertex of the corresponding triangle
$T_w=F_w(T),\; w=(w_1,w_2,\dots, w_n) \in S^n$ and can be represented
as
$$
\bigcap_{k=1}^\infty
F_{ w\underbrace{ \scriptscriptstyle iii\dots i}_\text{\normalfont $k$ times}}\,(T).
$$
Thus, for each node in $V_n$ we define an itinerary $w\bar{i}$,
where $\bar{i}$ stands for the infinite sequence of $i$'s: $iii\dots,$ $i\in S$.
Note that for $v\in V_n\setminus V_0$ there are two possible itineraries, e.g.,
$v_{1\bar{2}}$ and $v_{2\bar{1}}$
correspond to the same node from $V_2$.
\begin{figure}[h]
	\centering
	\includegraphics[width = .6\textwidth]{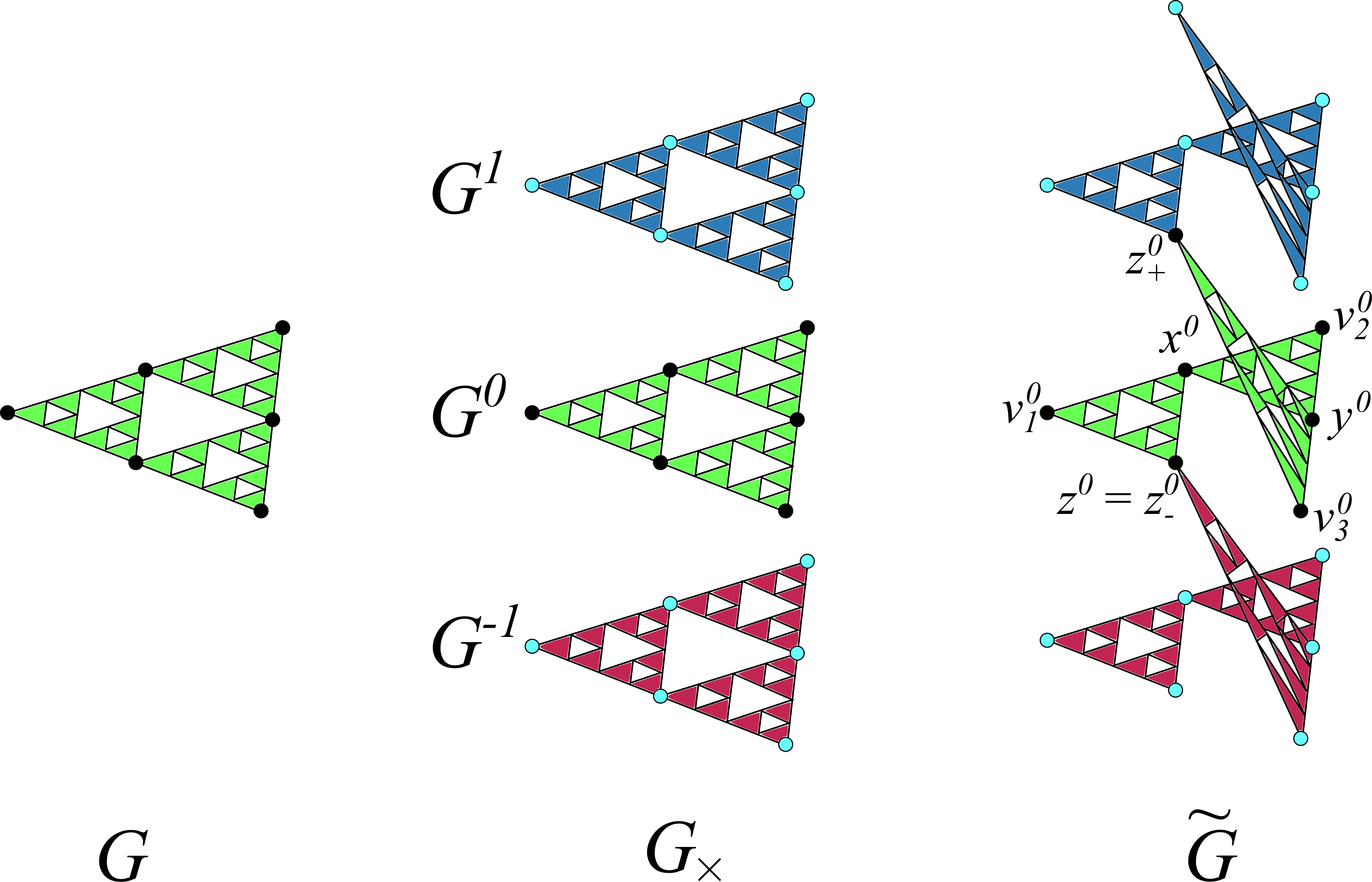}
	\caption{Construction of the covering space of $G$ corresponding to the degree $\bar{\omega}(f) = (1)$.}
	\label{fig:sg2}
\end{figure}
\begin{figure}[h]
	\centering
	\includegraphics[width = .4\textwidth]{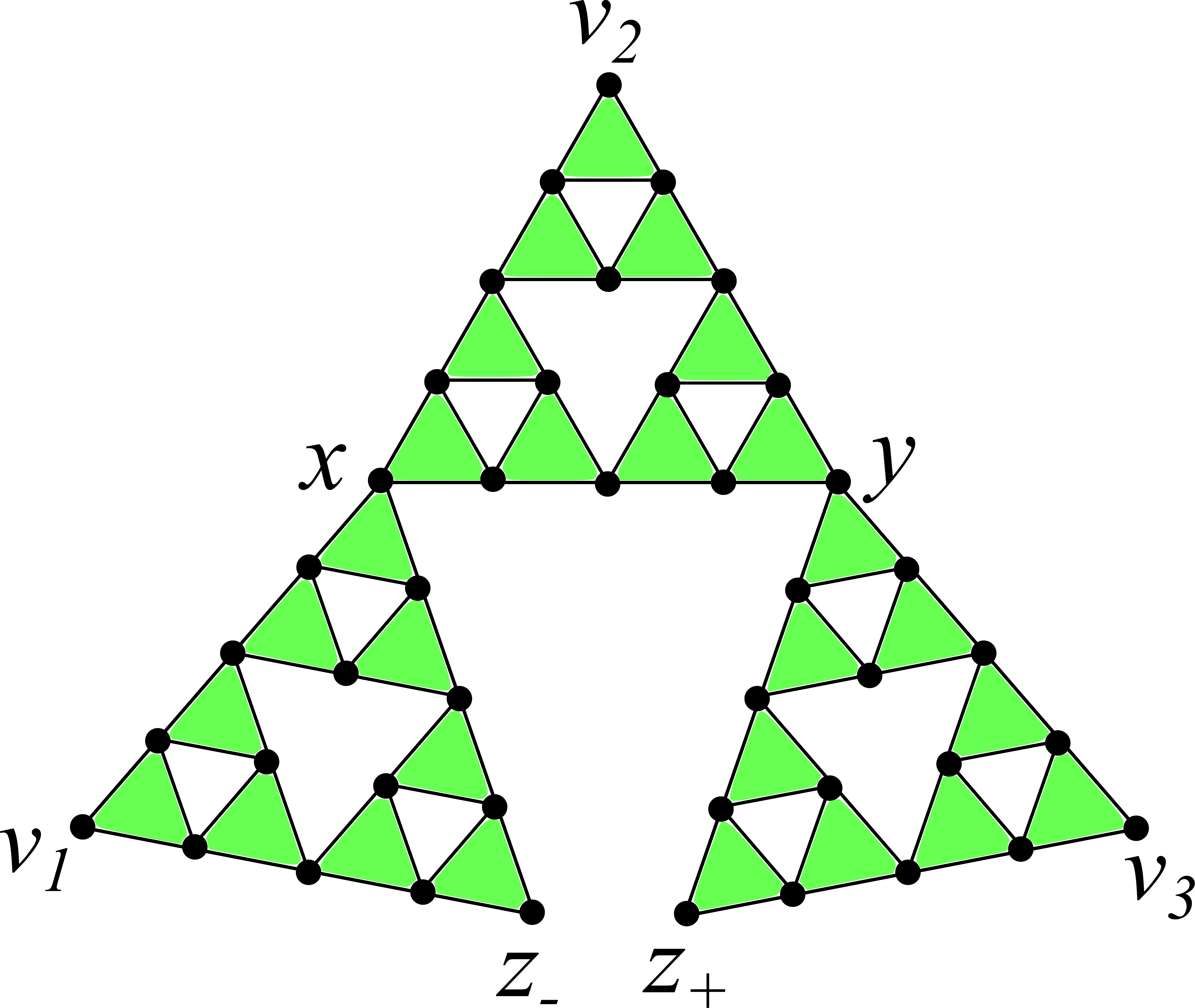}
	\caption{Approximating graphs $\Gamma_m^k$ for each sheet $G^k$ of the covering space $\tilde{G}$. Since our construction results in two copies of the vertex $z$, the resultant graphs are distinct from Figure \ref{fig:sg_graph}. Superscripts $k$ are suppressed for simplicity.}
	\label{fig:sg_graph_cut}
      \end{figure}
The construction of the covering space $\tilde G$ involves the 
following steps.
\begin{enumerate}
    \item
First, let
\begin{align}
    G_\times & \doteq G\times \Z, \label{G-stack}\\
    G^k & \doteq G\times \{k\} \subset G_\times, \quad k\in\Z, \label{k-sheet}\\
   G^k_i &\doteq F_i(G^k),\quad i\in S,\quad k\in \Z\\
    G^k_1\cap G^k_2  =\{x^k\}, &\quad  G_2^k\cap G_3^k =\{y^k\}, \quad  G_3^k\cap 
    G_1^k =\{z^k\}, \quad k\in\Z, \label{xyz}
\end{align}
(see Figure~\ref{fig:sg2}).
\item Then cut each $G^k$ at $z^k$, i.e., we replace $z^k$ with two distinct copies 
\begin{equation}\label{cut-z}
    z^k_- = v^k_{1\bar{3}}\quad\mbox{and}\quad
    z^k_+ = v^k_{3\bar{1}}.
\end{equation}
\item 
Identify
\begin{equation}\label{identify}
  z^k_+ =v_{3\bar{1}}^k \simeq v_{1\bar{3}}^{k+\rho_0} = z^{k+\rho_0}_-, \qquad k\in \Z.
\end{equation}
\item The covering space $\tilde G\doteq G_\times/^\simeq$, is the topological space obtained 
after identification \eqref{identify}.
  The copies of $G$, belonging to different levels $k\in\Z$, compose the sheets of $\tilde G$.
  We keep denoting them by $G^k$.
  Each sheet contains both copies of $z^k:$ $z_-^k$ and $z^k_+.$
  $G^0$ is called the \textit{fundamental domain}.
\item For $k\in\Z$,  we introduce a family  of graphs $\Gamma_m^k$ approximating $G^k$ (see Figure~\ref{fig:sg_graph_cut}). $\Gamma_m^k$ are constructed in the same way as graphs 
$\Gamma_m$  in the previous section (see Figure~\ref{fig:sg_graph}), with the only 
distinction that $z^k$ replaced with the 
the two copies $z_k^-$ and $z_k^+$ (see Figure~\ref{fig:sg_graph_cut}). 
  By identifying $V_m^k\ni z_+^k\simeq z_-^{k+\rho_0}\in V_m^{k+\rho_0}$, 
  we obtain the discretization
  of the covering space $\tilde G$, $\tilde \Gamma_m, \; m\in\N$. The set of nodes
  of $\tilde\Gamma_m$ is denoted by $\tilde V_m$, $\tilde V_\ast=\bigcup_{m=0}^\infty \tilde V_m$,
  and $\tilde V_\ast^k =\bigcup_{m=1}^\infty V_m^k.$
  $\tilde V_\ast$ is dense in $\tilde G.$
\end{enumerate}

In conclusion, we emphasize that the covering space is
constructed separately for each winding vector. For instance, when
$\rho_0=2$,  $G^k$ connects to $G^{k+2}$
(see Figure~\ref{fig:sg_graph_cut}) instead of $G^{k+1}$ for $\rho_0=1$ as shown in  Figure~\ref{fig:sg2}. The dependence of the covering space on 
the degree vector is the key feature of our approach.

      \begin{figure}[h]
	\centering
	\includegraphics[width = .15\textwidth]{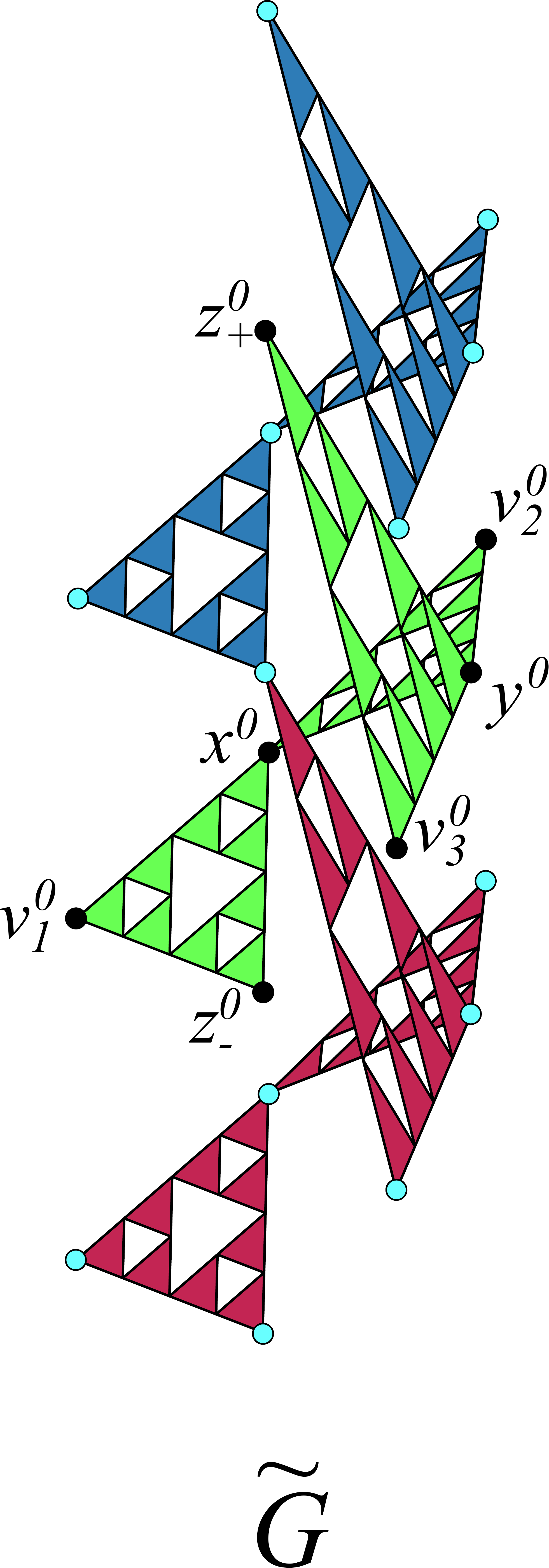}
	\caption{The covering space corresponding to $\rho_0=2$ .}
	\label{fig:covering-2}
\end{figure}

\section{Harmonic structure on $\tilde G$}\label{sec.harm-struct}
\setcounter{equation}{0}


 Harmonic functions on $\tilde G$ will be constructed via minimization of  the energy form to be defined next.
 To this end, let $\delta_1, \delta_2\in\R,$ and $\rho_0\in\Z$ be arbitrary fixed numbers and define
  \begin{equation}\label{def-H}
  H^k_m\doteq \{ f\in L(V_m^k, \R):\quad f(v_1^k)=k,\; f(v_2^k) = \delta_{1}+k,\; 
  f(v_3^k)= \delta_2+\delta_1+k,\;
  f(z_+^k) = f(z_-^k)+\rho_0\}
\end{equation}
for $m\in\N$ and $k\in\Z$ (see Figure~\ref{fig:sg_fund_domain}).

The energy form  on $\Gamma^k_m$ is defined as follows
  \begin{equation}\label{EGamma}
      E[\Gamma_m^k](f) = \left(\frac{5}{3}\right)^m\sum_{\xi\eta\in E(\Gamma_m^k)} \left(f(\eta)-f(\xi)\right)^2,
      \qquad f\in L(V_m^k,\R), m\in\N.
  \end{equation}

The energy form $E[\Gamma^0_m]$ inherits the key properties of the 
Dirichlet form $\cE_m$ (see \eqref{variational-harmonic}, \eqref{variational-extension}):
\begin{enumerate}
    \item A $\Gamma_m^0$-harmonic $f_m^0\in H_m^0$ minimizes $E[\Gamma^0_m]$ over $H_m^0$: 
    \begin{equation}\label{E-harmonic}
     E[\Gamma^0_m](f^0_m)=
    \min\left\{ E[\Gamma^0_m](f):\; f\in H_m^0 \right\}.
    \end{equation}
\item The minimum of the energy form $E[\Gamma^0_m]$ over all extensions of 
$f\in H_{m-1}^0$ to $H_m^0$ is equal to $E[\Gamma^0_{m-1}](f)$:
\begin{equation}\label{E-extension}
 \min\left\{ E[\Gamma^0_m](\tilde f):\; \tilde f\in H_m^0,\; 
 \tilde f|_{V_{m-1}}=f\in H_{m-1}^0 \right\} = E[\Gamma^0_{m-1}](f).
\end{equation}
\end{enumerate}

\begin{figure}[h]
	\centering
	\includegraphics[width = .3\textwidth]{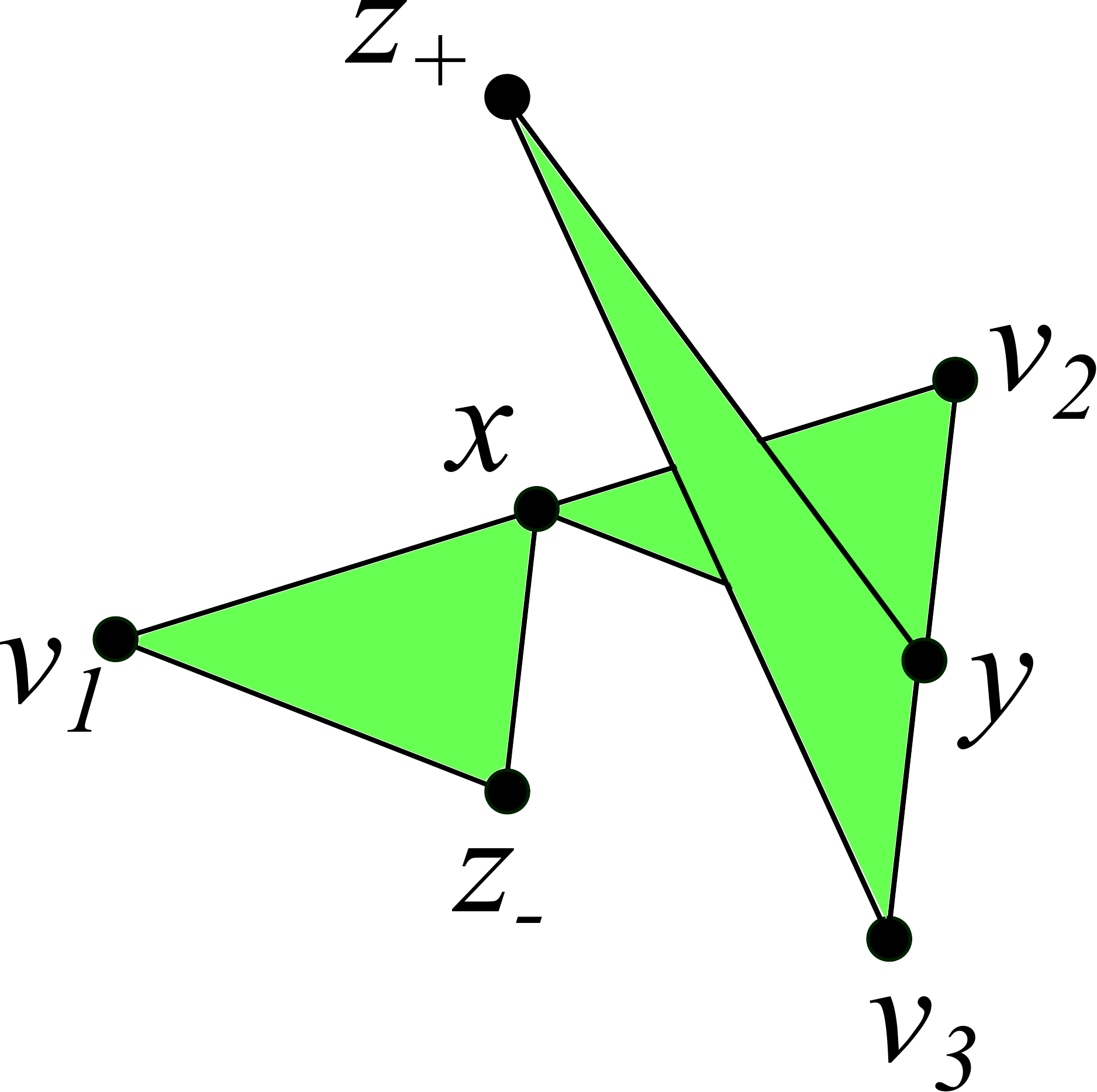}
	\caption{The boundary conditions for the minimization problem \eqref{variational}:
 the values of the function are imposed at  $v_1$, $v_2$, $v_3$ and the  jump condition - at $z_+$. }
	\label{fig:sg_fund_domain}
\end{figure}

For every $m\in\N$, consider the minimization problem 
\begin{equation}\label{variational}
    E[\Gamma_m^0](f)\longrightarrow \min_{f\in H_m^0}.
\end{equation}

\begin{lemma}\label{lem.minimizer}
    The variational problem \eqref{variational} has a unique solution
    $f^0_m \in H^0_m$:
    $$
    E[\Gamma_m^0](f^0_m)=\min_{f\in H_m^0} E[\Gamma_m^0](f).
    $$
    Moreover,
    \begin{enumerate}
    \item $f_m^0$ \; is $\Gamma^0_m$-harmonic,
        \item $f^0_{m+1}|_{V_m} = f_m^0$.
    \end{enumerate}
\end{lemma}
\begin{proof}
Let $m\in\N$ be arbitrary but fixed. 

Since $E[\Gamma_m^0](f)$ is bounded from below and the sublevel sets
$\{ f:\;E[\Gamma_m^0](f)\le c\}$ are compact, we conclude that \eqref{variational} has 
at least one solution.

Let $f^\ast$ be a minimizer of $E[\Gamma_m^0](f)$. By \eqref{E-extension}, 
$E[\Gamma_m^0](f^\ast)=E[\Gamma_1^0](f^\ast|_{V_1^0})$.

Thus, we first find the values of $f^\ast$ at points $x, y,$ and $z:=z_-$  such that
$$
E[\Gamma_1^0]\longrightarrow \min_{f_x, f_y, f_z}.
$$
This yields
\begin{equation}\label{Hext}
\begin{pmatrix} f_x \\ f_y\\ f_z 
\end{pmatrix}
=
\begin{pmatrix}
    \frac{2}{5} &  \frac{2}{5} &  \frac{1}{5}\\
     \frac{1}{5} &  \frac{2}{5} &  \frac{2}{5}\\
      \frac{2}{5} &  \frac{1}{5} &  \frac{2}{5}
\end{pmatrix}
\begin{pmatrix} f_{v_1} \\ f_{v_2}\\ f_{v_3}
\end{pmatrix}
+
\begin{pmatrix}
   \frac{3}{10} &  \frac{1}{10} &  \frac{1}{10}\\ 
    \frac{1}{10} &  \frac{3}{10} &  \frac{1}{10}\\ 
     \frac{1}{10} &  \frac{1}{10} &  \frac{3}{10}
\end{pmatrix}
\begin{pmatrix} 0\\ \rho_0\\ -2\rho_0 .
\end{pmatrix}
\end{equation}
Having found $f_x,$ $f_y,$ and $f_z$, we now use the harmonic extension \eqref{classical-extension} in the 
subdomains $v_1xz$, $xv_2z$, and $zyv_3$. 
The uniqueness of the minimizer follows by construction.
Properties 1. and 2. stated in the lemma follow from 
 \eqref{variational-harmonic}-\eqref{variational-extension}.
\end{proof}

\begin{remark}
    The first term on the right-hand side of \eqref{Hext} coincides with the 
    standard harmonic extension (cf.~\eqref{classical-extension}).
    The second term accounts for the jump condition $f(z_+^0)=f(z_-^0)+\rho_0$.
  In \cite{Strich02},  the extension formula \eqref{Hext} was derived without
    the use of the covering space.
\end{remark}

Having found the minimizer $f^\ast$ on $V^0_\ast$, we extend it by continuity to the fundamental domain $G^0$,
and further to a uniformly
continuous function on the covering space:
\begin{equation}\label{k-period}
\mathbf{f}^\ast (x,k) = f^\ast(x) +k , \qquad x\in G^0, \; k\in \Z.
\end{equation}

\begin{theorem}\label{lem.ext_harmonic_on_G}
  The restriction of $\mathbf{f}^\ast$ on $\tilde V_m$, $\mathbf{f}^\ast|_{\tilde V_m}$, is $\tilde\Gamma_m$-harmonic
  for every $m\in \N$:
 \begin{equation}\label{G-harmonic}
  \Delta_m \mathbf{f}^\ast|_{\tilde V_m}(x) = 0\quad \forall x\in \tilde{V}_m.   
  \end{equation}
\end{theorem}
\begin{proof}
  By construction, $\mathbf{f}^\ast|_{V_m^k}$ is the unique minimizer of
$
E[\Gamma^k_m] 
$
on $H_m^k$. Thus, $\mathbf{f}^\ast|_{V_m^k}$ is $\Gamma^k_m$-harmonic and
it remains to check that $\mathbf{f}^\ast|_{\tilde V_m}$ satisfies the 
discrete Laplace equation at $z_{\pm}^k$. It is sufficient to
verify this only at $z_{-}^0$.

To simplify the notation, for the reminder of the proof we set $r:=\rho_0$.
Denote the
graph obtained by gluing $\Gamma_m^0$ and $\Gamma_m^{-r}$ by
$$
\Gamma_m=\Gamma_m^0 \underset{z_+^{-r}\simeq z_-^0}{\bigsqcup} \Gamma_m^{-r}.
$$
where for the vertex set of $\Gamma_m$ we take  the union of the node sets
of these two graphs with $z\doteq z_+^{-r}\simeq z_-^0$ and keep the adjacency relations from
$\Gamma_m^0$ and $\Gamma_m^{-r}$. Denote the node set of $\Gamma_m$ by $V_m$.

Define $f_m\in L(V_m,\R)$ by
\begin{equation}\label{def-f}
  f_m(x)=\left\{ \begin{array}{cc}
                   f_m^0(x),& x\in V_m^0,\\
                   f_m^{-r}(x), & x\in V_m^{-r}.
                 \end{array}
               \right.
             \end{equation}
             We want to show that $f_m$ is the minimizer of
             \begin{equation}\label{new-min}
               E[\Gamma_m](f)\longrightarrow \min_{\phi \in H_m},
             \end{equation}
             where
             $$
             H_m=\{ \phi\in L(V_m,\R):\; \phi(v_1^0)=0, \; \phi(v_2^0)=\delta_1,\; \phi(v_3^0)=\delta_1+\delta_2,
             $$
             $$
             \phi(v_1^{-r})=-r,\;  \phi(v_2^{-r})=\delta_1-r,\; \phi(v_3^{-r})=\delta_1+\delta_2-r, \; \phi(z_+^0)=\phi(z_-^{-r})+2r\}.
             $$
             This will imply that $f_m$ is $\Gamma_m$-harmonic at $z$, because $z$ is an interior point of $\Gamma_m$.

             First, note that minimization problem \eqref{new-min} has a unique solution (cf.~Lemma~\ref{lem.minimizer}),
             which we denote by
             \begin{equation}\label{def-psi}
             \psi =\left\{ \begin{array}{cc}
                             \psi_0,& x\in V_m^0,\\
                             \psi_1,& x\in V^{-r}_m.
                           \end{array}\right.
                         \end{equation}
                         Using the symmetry of the energy form, observe that if \eqref{def-psi} minimizes \eqref{new-min} then so does
                        \begin{equation*}
             \tilde\psi =\left\{ \begin{array}{cc}
                             \psi_1+r,& x\in V_m^0,\\
                             \psi_0-r,& x\in V^{-r}_m.
                           \end{array}\right.
                         \end{equation*}  
  The uniqueness of the minimizer \eqref{new-min} then implies
                           $$
                           \psi_0(z^0_+)-\psi_0(z_-^0)=\psi_1(z^{-r}_+)-\psi_0(z^{-r}_-)=r.
                           $$
                           From this we conclude that      $\psi|_{V^0_m}\in H_m^0$ and $\psi|_{V^{-r}_m} \in H_m^{-r}$.
                           Thus, $\psi=f_m$ (cf., \eqref{def-f}), which means that $f_m$ is $\Gamma_m$-harmonic.
\end{proof}

\section{Simple harmonic maps from SG to $\T$}
\label{sec.HM}
\setcounter{equation}{0}

Having constructed $\mathbf{f}^\ast$, a harmonic function on the covering space, we are one step away 
from producing
a HM on SG. This is done by restricting the domain of $\mathbf{f}^\ast$ to the fundamental domain 
and by projecting the
range of $\mathbf{f}^\ast$ to $\T$: 
\begin{equation}\label{project-back}
    \hat f^\ast\doteq \mathbf{f}^\ast|_{G_0}\mod 1.
  \end{equation}

  By construction $\hat f^\ast$ satisfies the boundary conditions \eqref{bc-hm}.
  By Theorem~\ref{lem.ext_harmonic_on_G}, for every $m\in\N$, the restriction of
$\hat f^\ast$ to the triangular mesh $\Gamma_m^0$,
$\hat f^\ast_m\doteq \hat f^\ast\left|_{\Gamma_0^m}\right.$ is $\Gamma_m$-harmonic:
\begin{equation}\label{local-harm}
  \sum_{y:y\sim_m x}\left[ \hat f^\ast_m(y)- \hat f^\ast_m(x)\right]\mod 1 =0,\quad
  x\in V_m^0\setminus \{v_1, v_2, v_3\}.
\end{equation}
Since \eqref{local-harm} holds for every $m\in\N$, $\T$-valued $\hat f^\ast$ is harmonic at every
$x\in V^\ast\setminus V^0$ and satisfies boundary conditions \eqref{bc-hm}.

Next we give a variational interpretation of $\hat f^\ast$.
To this end, let $\delta_1, \delta_2\in \R$ and $\rho_0\in \Z$ are the same as above and denote
\begin{equation}\label{H-hat}
    \hat H=\{ f\in C(G,\T):\; \p_{v_1v_2} f=\delta_1, \; \p_{v_2v_3} f=\delta_2,\; 
    \p_{v_3v_1} f=\rho_0-\delta_1-\delta_2\}
\end{equation}
(see \eqref{bc-hm} for the definition of $\partial_{v_iv_j}$). 

Below we show that $\hat f^\ast$ minimizes 
\begin{equation}\label{Em-hat}
    \hat\cE_m(f)=\left(\frac{5}{3}\right)^m\sum_{xy\in E(\Gamma_m)}d_\T\left(f(x), f(y)\right)^2
  \end{equation}
over $f\in \hat H$.   Here, $d_\T(\cdot,\cdot)$ stands for the geodesic distance on $\T$. 

Let $\hat f\in \hat H$.
Since $\hat f$ is uniformly continuous on $G$, for sufficiently large 
$m$,
\begin{equation}\label{coincide}
 \hat\cE(f)= E[\Gamma_m^0](f),\quad m\ge m_0\in\N,
 \end{equation}
 where $f$ is the real-valued lift of $\T$-valued $\hat f$.

 Since $E[\Gamma_m^0](f)$ is nondecreasing in $m$, $\hat\cE_m(f)$ is nondecreasing
 for $m\ge m_0$. Thus,
 \begin{equation}\label{hat-cE-min}
\hat\cE(f)=\lim_{m\to m} \hat\cE_m(f) 
\end{equation}
 is well-defined. Using \eqref{coincide} and Lemma~\ref{lem.minimizer}, we conclude that
 $\hat f^\ast$ is the minimizer of $\hat\cE$ over $\hat H$.

\section{Higher order maps}\label{sec.higherorderSG}
\setcounter{equation}{0}
In the previous section, we have completed the construction of HMs with
simple degree. In this section, we extend the algorithm for the case of 
an arbitrary degree.

Let $f: G\to\T$ be a HM of degree 
\begin{equation}\label{deg-f}
  \bar\omega(f) = (\rho_0,\rho_1,\rho_2,\rho_3, \dots)\in\Z^\N.
\end{equation}
We say that $f\in C(G,\T)$ is a HM of order $N\in\N$ if
\begin{align*}
  \rho_i\neq 0, & \quad\mbox{for some}\quad \frac{3^N-3}{2}< i\le \frac{3^{N+1}-3}{2},\\
  \rho_j=0, & \quad\mbox{for all}\quad  j> \frac{3^{N+1}-3}{2}.
\end{align*}

The construction of the covering space for higher order HMs requires additional cuts, which we 
explain below. To this end, note that the degree of the HM $f$ of order $N$ can be rewritten as
\begin{equation}\label{re-deg}
  \bar\omega(f) = \left(\rho_{\ell(w)}, \; |w|\le N\right),
\end{equation}
where
\begin{align*}
    \ell(w)=0, & \qquad w=\emptyset,\\
    \ell(w)=\sum_{i=1}^m w_i\cdot 3^{i-1}, &\qquad  w=(w_1, w_2, \dots, w_m)\in S^m.
\end{align*}

For $w=\emptyset$, we choose the same cut points as before
$$
\xi^k_0\doteq v^k_{1\bar{3}}=v^k_{3\bar{1}}, \qquad k\in\Z.
$$
and identify 
$$
v^k_{1\bar{3}}\simeq v^{k+\rho_0}_{3\bar{1}}. 
$$
At step $m>0$, we have $0<|w|=m\le N$.
For each $w=(w_1,w_2, \dots, w_m)$, let $\gamma_w$ stand for the boundary of the triangular
cell $T_w$. Recall the reference points $\xi_{\gamma_w}$, which were defined in Section~\ref{sec.maps}
(see Figure~\ref{fig:f-cuts}). These points will serve as the cut-points at each level.

Denote the two itineraries of $\xi^k_{\gamma_w}$ by $c_-(\xi^k_{\gamma_w})<c_+(\xi^k{\gamma_w})$.

Next, we cut at $\xi_{\gamma_w}^k, \; k\in\Z,$ and identify
\begin{equation} 
	\label{eq:sg_gen_cuts}
v^k_{c_-(w)}\simeq v^{k+\ell(w)}_{c_+(w)}.
\end{equation}
The cut points for the first, second, and third order HMs are shown in Figure~\ref{fig:f-cuts} 
(also see Appendix \ref{appendix} for explicit formulae for the cut points needed for
construction of the second order HMs). 

The remainder of the algorithm proceeds as in Section~\ref{sec.harm-struct}. Specifically, 
let $\tilde G=G_\times/^\simeq.$ Then approximate the sheets of $\tilde G,$ $G^k,$ by graphs
$(\Gamma^k_m, \; k\in \Z, \; m\in \N).$ Define discrete spaces
$$
\begin{array}{ll}
H_m^k=\left\{ \right. f\in L(V_m^k),\R): &f(v_1^k)=k,\; f(v_2^k)=\delta_{1}+k,\; f(v_3^k)=\delta_{2}+\delta_1+k,\\
 & f(v^k_{c_+(w)})=f(v^k_{c_-(w)})+\rho_{\ell(w)},
\; |w|\le N \left. \right\}.
\end{array}
$$

As in Lemma~\ref{lem.minimizer}, one can show that for every $m\in\N$
$$
E[\Gamma^0_m]\longrightarrow \min_{f\in H^0_m}
$$
has a unique solution, $f^0_m,$ which can be computed via harmonic extension.
By extending this solution to $\tilde G$ as in \eqref{k-period}, we obtain a 
$\tilde\Gamma_m$-harmonic function. By repeating this procedure for all $m\in\N$,
we obtain the values of harmonic $\mathbf{f}\in L(\tilde G, \R)$, which after restricting to the fundamental domain and projecting to $\T$ yields the desired
HM with the prescribed degree \eqref{def-f}. Figure~\ref{fig:sg1}{\bf d} shows an example of a HM with degree $\bar{\omega}(f)=(1,1,1,1)$.

\section{Extending to p.c.f. fractals}\label{sec.pcf}
\setcounter{equation}{0}

  In the previous section, we used SG to develop a method for constructing $\T$-valued HMs on a fractal.
  While the SG was a convenient example, the method itself is not restricted to the specific geometry of the SG.
  It is the goal of this section to clarify the scope of our method.

  We begin by reviewing the main ingredients of the method at hand.
  \begin{description}
  \item[(a)] The covering space is used to accomodate the global topological restrictions of a given
    homotopy class encoded in the degree vector.
  \item[(b)] The harmonic extension algorithm is used to compute real-valued harmonic function subject to
    given boundary conditions.\footnote{Strictly speaking, we only need existence and uniqueness of solution
      of the minimizer of $\hat\cE\rightarrow \min_{\hat H}$ (cf.~\eqref{hat-cE-min}) or existence and uniqueness of
      solution of the boundary value problem \eqref{Laplace}, \eqref{bc-hm}.}
  \end{description}

  Below, we examine the conditions on fractal sets required for
  assumptions (a) and (b) above.
  In addition to compactness and finite ramification, which were
  discussed previously,
  it is essential that the fractal supports a harmonic structure, as
  required for part (b). The natural class
  of fractals for which one may expect the main result of this paper
  to hold is that of p.c.f. fractals \cite{Kig01}.
  P.c.f. fractals were introduced by Kigami as a class of domains on
  which the Laplacian can be defined as
  the limit of a sequence of appropriately rescaled graph Laplacians,
  in direct analogy with the
  construction for the SG. The existence of a harmonic structure on a
  p.c.f. fractal is not automatic;
  rather, it depends on the solvability of a nonlinear eigenvalue
  problem, referred
  to as the renormalization problem \cite{Kig01}. For many important
  examples
  this renormalization problem can be solved, making p.c.f. fractals a
  rich source of domains on which analysis can be systematically
  developed.
  For this reason, we adopt p.c.f. fractals as the primary class of
  domains
  for the construction of HMs to the circle.

  Before we give the definition of a p.c.f. fractal $K$ and and formulate a set of assumptions on $K$, we explain
  the nature of these assumptions.
  \begin{description}
  \item[(0)] General topology. As an attractor of an iterated function system $K$ is automatically compact.
Further, as a p.c.f. set $K$ is finitely ramified (cf.~\cite{Str06}).
   We will however need to assume in addition that $K$ is connected.
  \item[(1)] Harmonic structure. Our key assumption on $K$ is the existence of self-similar compatible sequence
    of discrete Dirichlet forms (see \eqref{self-similarity}, \eqref{compatibility}). Harmonic extension relies 
    on this assumption.
  \item[(2)] Cycle space. One important distinction of the SG from a general p.c.f. fractal is the natural hierarchy
    of nested triangular loops $\gamma_0, \gamma_1, \gamma_2,\dots$,
    which span the cycle space of each
    $\Gamma_m$ and whose union yields a basis of the cycle space of
    the fractal skeleton $\cup_{m=0}^\infty \Gamma_m$.
    For a general p.c.f. fractal,
    it will not be immediately obvious how to choose a basis at a given level.
    
    To overcome this problem, we first identify $N$, the lowest level that contains all
    loops with nontrivial winding number. Such $N$ exists due to uniform continuity.
    Next, we fix a spanning tree of $\Gamma_N$, which will in turn determine the
    corresponding basis of the cycle space.

    The chosen basis needs to satisfy certain properties, which 
    guarantee that we can further choose appropriate cut points. Below, we formulate these
    conditions as a set of assumptions on a fractal set.
     \end{description}


  The remainder of this section is organized as follows. In \S~\ref{sec.pcf-assumptions}, we review the
  definition of a p.c.f. fractal and a
  harmonic structure on the p.c.f. fractals following \cite{Kig01} (see \cite{Str06, Saenz-HarmAnal} for
  a concise elementary introduction to p.c.f. fractals).
Subsequently, in \S~\ref{sec.pcf-main}, we discuss specific features of the SG that are absent in a general
p.c.f. set. These differences prompt certain adjustments in our approach to the general case. 
After discussing the setting and formulating the corresponding assumptions, we present 
Theorem~\ref{thm.pcf}, the main result of this section.
This is followed by a discussion of several representative examples in \S~\ref{sec.pcf-examples}.

\subsection{P.c.f. fractals and harmonic structure}\label{sec.pcf-assumptions}

Let $K\subset\R^d$ be an attractor of an iterated function system, i.e., $K$ is compact
set satisfying
\begin{equation}\label{IFS}
    K=\bigcup_{i=1}^N F_i(K),
\end{equation}
where $F_i:\R^d\to\R^d, \; i\in [N],$ such that 
$$
|F_i(x)-F_i(y)|= \lambda_i |x-y|, \qquad 0<\lambda_i\le \lambda<1.
$$
In addition, we assume that $K$ is connected.

Let $S=\{1,2,\dots, N\}$ and define $m$-cells as before
$$
K_w=F_w(K), \quad w=\left(w_1,w_2, \dots, w_m\right)\in S^m.
$$
The critical set of $K$ is defined by
$$
\mathcal{C}\doteq \bigcup_{i\neq j} (K_i\bigcap K_j).
$$
Assume that $\mathcal{C}\neq\emptyset$ and define 
\begin{equation}\label{pcf-boundary}
V_0\doteq \bigcup_{m\ge 1} \bigcup_{|w|=m} F^{-1}_w(\mathcal{C}).
\end{equation}
\begin{definition} Let $K$ be a compact connected set satisfying \eqref{IFS}.
$K$ is called a p.c.f. fractal if $V_0$ is finite.
\end{definition}

For the remainder of this section, we assume that $K$ is a p.c.f. fractal. 

Let 
$$
V_m=\bigcup_{i=1}^N F_i(V_{m-1}), \quad m\ge 1.
$$
Note that 
$$
V_m=\bigcup_{|w|=m} F_w(V_0),
$$
as in the case of SG discussed above. Further,
$$
V_\ast=\bigcup_{m\ge 0} V_m
$$
is dense in $K$.

Next, we we equip $K$ with a harmonic structure.
Specifically, we introduce a sequence of quadratic forms
$$
\mathcal{E}_m(u)=\frac{1}{2}\sum_{i,j=1}^N c^m_{ij} \left(u(j)-u(i)\right)^2, \quad u\in L(V_m,\R),
$$
where $(c^m_{ij})$ is a positive definite $m\times m$ matrix.

\begin{assumption}\label{as.harmonic}
The sequence $(\mathcal{E}_m)$ satisfies the following two conditions:
\begin{description}
    \item[i] Self-similarity. There exist $0< r_1, r_2, \dots, r_N<1$ such that
    \begin{equation}\label{self-similarity}
    \mathcal{E}_m(u)=\sum_{i=1}^N \frac{1}{r_j} \mathcal{E}_{m-1} (u\circ F_j), \quad u\in L(V_m,\R).
    \end{equation}
    \item[ii] Compatibility. For every $m\ge 1,$
    \begin{equation}\label{compatibility}
    \mathcal{E}_{m-1}(v)=\min\left\{  \mathcal{E}_m(u):\; u\in L(V_m,\R),\; u|_{V_{m-1}}=v\right\}.
    \end{equation}
  \end{description}
\end{assumption}

The compatibility condition implies that $\mathcal{E}_m(u|_{V_m})$ is 
nondecreasing for any $u\in C(K,\R)$.
Thus, 
$$
\mathcal{E}(u)\doteq \lim_{m\to\infty}\mathcal{E}_m(u|_{V_m}),\quad u\in \operatorname{dom}\mathcal{E}.
$$
is well-defined and 
$$
\operatorname{dom}\mathcal{E}=\{ u\in C(K,\R): \; \lim_{m\to\infty} \mathcal{E}_m(u|_{V_m})<\infty\}.
$$
Furthermore, thanks to \eqref{compatibility}, the minimization problem 
$$
\mathcal{E}(v)\longrightarrow \min_{v\in C(V_\ast,\R)}
$$
can be solved recursively
$$
v_m^\ast=\operatorname{argmin} \mathcal{E}_m\left( v|_{V_{m-1}}=v^\ast_{m-1}\right), \qquad m=1,2,\dots.
$$
This is the basis of the harmonic extension algorithm.

The harmonic structure implicitly defines a sequence of graphs $\Gamma_m=\langle V_m, \sim_m\rangle$
approximating $K$:
$$
x\sim_m y\quad\mbox{if}\quad c_{xy}^m>0, \quad x,y\in V_m.
$$

\begin{figure}[h]
 	 \centering
 	\includegraphics[width = .5\textwidth]{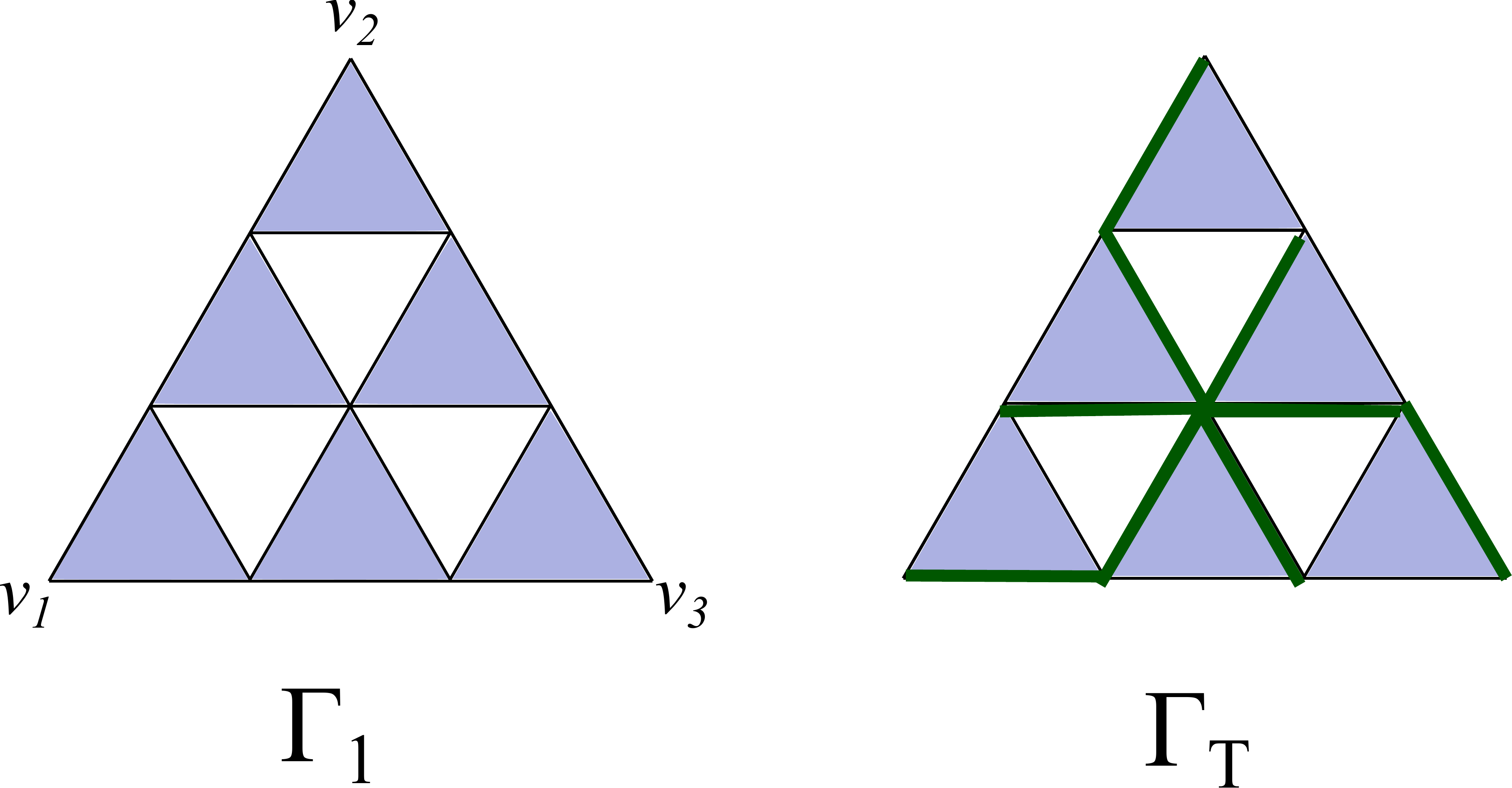}
 	\caption{Graph approximation $\Gamma_1$ for $SG_3$, and a spanning tree, $\Gamma_T$.}
 	\label{fig:sg3_tree}
 \end{figure}
 
  \begin{figure}[h]
 	\centering
 	{\bf a)}\includegraphics[width = .45\textwidth]{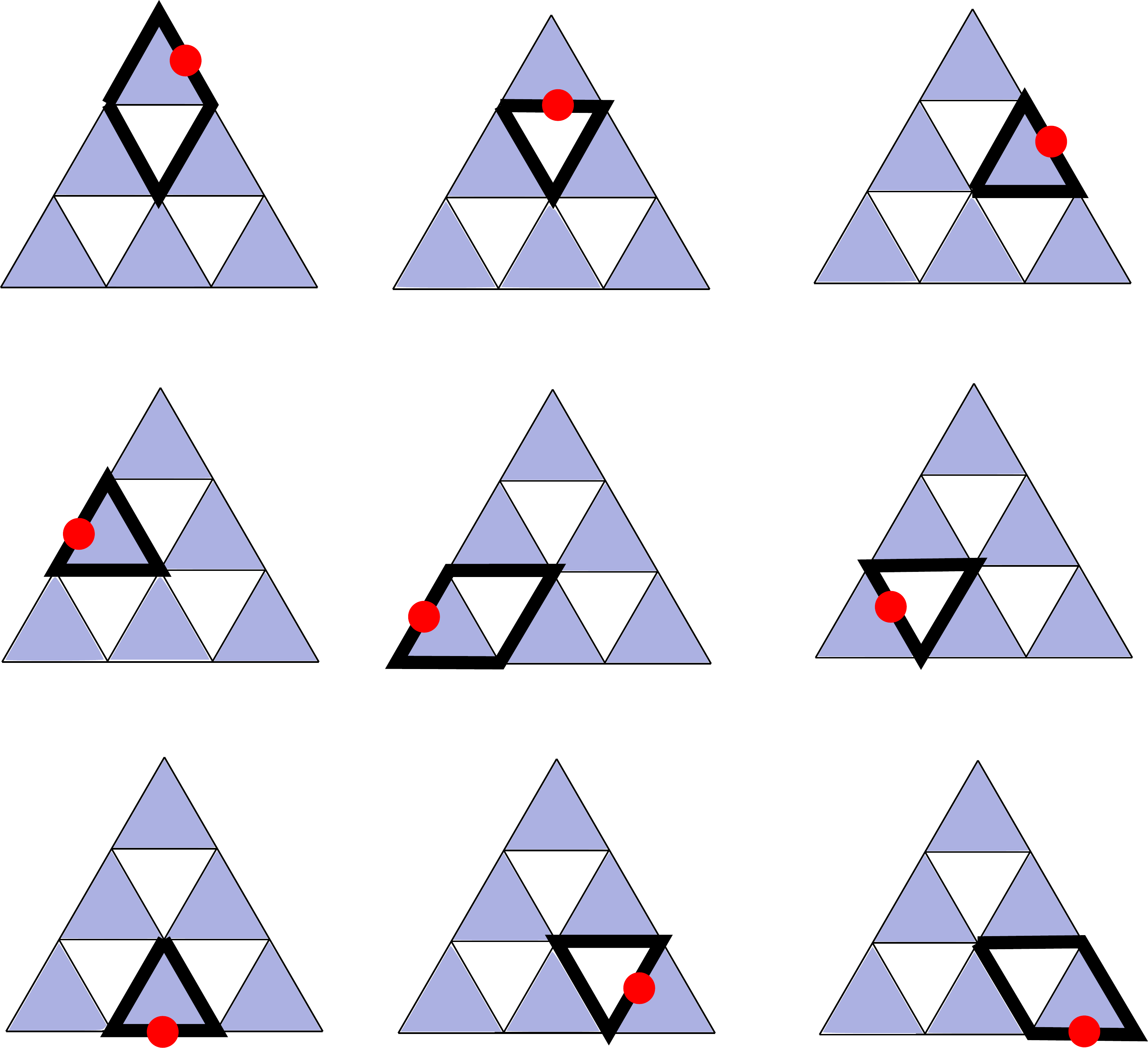}
 	\hspace*{1cm}
 	 	{\bf b)}	\includegraphics[width = .4\textwidth]{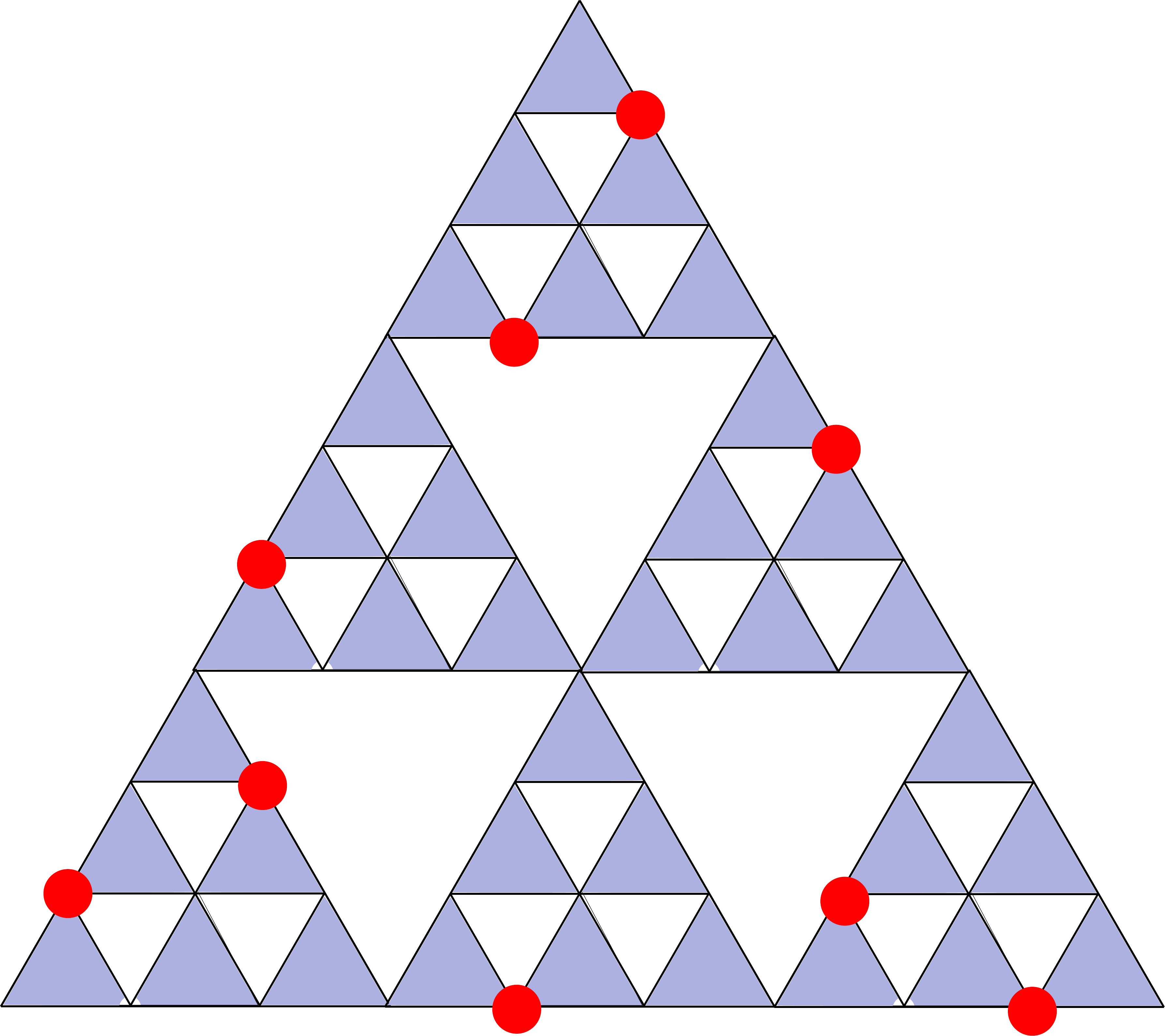}
                \caption{{\bf a)} The basis of the cycle space generated by $\Gamma_T$ in Figure \ref{fig:sg3_tree}.
                  Red points indicate where cut points should be chosen to satisfy Assumption \ref{as.cycle}.
                  {\bf b)} Cut points shown on $\Gamma_{2}$. Unlike $G$, there are non-unique choices for the cut points.}
 	\label{fig:sg3_basis}
      \end{figure}
      
\subsection{The cycle space and the main result}\label{sec.pcf-main}

  We now specify our geometric assumptions on the cycles of $K$. 
Before doing so, it is instructive to consider the $3$-level 
SG shown in Fig.~\ref{fig:sg3_tree}.

It is tempting, by analogy with the standard SG,
to conclude that a basis of the cycle space of $\Gamma_1$
consists of the outer boundary of the large triangle together
with the boundaries of the six smaller triangles.
However, this would be incorrect.
The spanning tree displayed in the right panel of
Fig.~\ref{fig:sg3_tree} shows that the cycle space is
six-dimensional.
The corresponding basis is presented in
Fig.~\ref{fig:sg3_basis}.

This example demonstrates that, in contrast to the situation
for the SG, one cannot in general rely on the
existence of a nested family of cycles to determine the basis of the cycle space.
Consequently, a basis for the cycle space must be constructed
individually in each case.


The second point that needs to be addressed is the relation between cycle
spaces at different levels of discretization. Specifically, let $m\in\N$ be the lowest level,
which contains all cycles with nontrivial winding numbers 
 and denote
    by $\Gamma_m$ the graph approximating $K$ at level $m\in\N$. Choose a basis for the cycle
    space of $\Gamma_m$, $\operatorname{Cyc}(\Gamma_m)$:
\begin{equation}\label{basis-m}
\mathcal{B}_m=\left\{\gamma_1^{(m)}, \gamma_2^{(m)},\dots, \gamma_{n_m}^{(m)} \right\}.    
\end{equation}

Let $\gamma\in\cB_m$ be such that $\omega_\gamma(f)\neq 0$. Thus, $\gamma$ needs to be
included in the construction of the covering space. In particular, we need to choose a cut point on
$\gamma$ as we have done for the SG. The problem is that the cut point must be chosen from the
vertices at the next level of discretization $V_{m+1}\setminus V_m$. In case of the SG, we were able
to formulate a simple rule how to choose these points at any given level. For a  general p.c.f. fractal
the situation is more interesting. For instance, in case of the pentagasket shown in
Figure~\ref{fig:penta_cuts}, the cut points  do not belong to the original cycles, but to their suitably
chosen counterparts at the next level of discretization. Since there does not seem to be any
canonical way of embedding the cycles from the basis of the cycle space to the cycle space at the next
level, we postulate such an embedding and require that it  satisfies certain conditions. We verify
these conditions in each example of p.c.f. fractals below.

For the next assumption we will need the following notation.
Let $V(\gamma)\subset V_m$ stand for the set of vertices belonging to $\gamma\in\cyc{\Gamma_m}$.

\begin{assumption}\label{as.embed}
  Let $\gamma\in\cB_m$ then there is a mapping $\iota:\cB_m\rightarrow\cyc{\Gamma_{m+1}}$ such that
  \begin{enumerate}
  \item $V(\gamma)\subset V\left(\iota(\gamma)\right)$,
    \item $\omega_\gamma(f)=\omega_{\iota(\gamma)}(f)\quad \forall \gamma\in\cB_m.$
    \end{enumerate}
  \end{assumption}

  \begin{assumption}\label{as.cycle}
    We assume that there are points $\xi^m_k\in V_{m+1}\setminus V_m$ such that
    $\xi_k^m\in \iota\left(\gamma_k^{(m)}\right)$
    but $\xi_k^m\notin \iota(\gamma_l^{(m)})$ for every $l\neq k$ and
    $k\in [n_m]$.
\end{assumption}

Next, we turn to the boundary value problem for HMs on a general p.c.f. fractal.
Following the strategy developed for the SG, to construct HMs on p.c.f. 
fractals we formulate the boundary value problem for a harmonic function 
on the fundamental domain of the covering space.
To this end, fix level $m\in\N$ and the degree vector 
\begin{equation}\label{eq:winding_vec_general}
  \bar\omega^{(m)}(f)=\left( \omega_{\gamma_1^{(m)}}(f),\dots,
    \omega_{\gamma_{n_m}^{(m)}}(f)\right)=(\rho_1,\dots,\rho_{n_m}),
\end{equation}
where 
$
\omega_{\gamma}(f)
$
is the degree of $f$ restricted to $\gamma$.

Let $\gamma^{(m)}_0\in\operatorname{Cyc}(\Gamma_m)$ be a cycle that contains all boundary vertices,
$V_0\subset V(\gamma_0)$. We express $\gamma^{(m)}_0=\sum_{i=1}^{n_m} a_i\gamma_i^{(m)}$
where $a_i\in\Z$. Then
$$
\omega_{\gamma_{0}^{(m)}}(f)=\sum_{i=1}^{n_m}a_i\rho_i.
$$

We are now prepared to formulate the boundary condition (see \eqref{bc-hm}):
\begin{align}
    \p_{v_1v_2} f&= \delta_1, \nonumber\\
    \p_{v_2v_3} f&= \delta_2, \label{pcf-bc-hm}\\
    &\dots\nonumber\\
    \p_{v_{\ell-1}v_{\ell}} f &= \delta_{\ell-1}\nonumber \\
    \p_{v_{\ell}v_1} f&=\rho_0-\delta_1-\delta_2-\dots-\delta_{\ell-1}, \nonumber
\end{align}
where $\rho_0=\omega_{\gamma_{0}^{(m)}}(f)$ and $V_0=\{v_1,v_2,\dots, {v_\ell}\}.$

\begin{theorem}\label{thm.pcf}
  Suppose $K$ is a connected p.c.f. fractal and Assumptions~\ref{as.harmonic}-\ref{as.cycle} hold.
  Then there is a unique HM $f\colon K\to\mathbb{T}$ satisfying 
  \eqref{eq:winding_vec_general} and \eqref{pcf-bc-hm}.
\end{theorem}

\begin{proof}
  The key constructions and the corresponding results in  
  Sections~\ref{sec.cover}-\ref{sec.higherorderSG} translate to the present
  setting with minor modifications, so we highlight only the main distinctions from the previous
  proof for SG.
  
  First, denote the trivial covering space $K_\times = K\times \mathbb{Z}$.  
  Using Assumption \ref{as.embed}, embed the cycles $\gamma^{(m)}_i$ into
  $\operatorname{Cyc}(\Gamma_{m+1})$.  Using Assumption \ref{as.cycle}, for each $i\in [n_m]$, fix a cut point 
  $v_i^{(m)}=\xi(\gamma_i^{(m)}) \in V_{m+1}$ disjoint from the other embedded  cycles 
  ${\gamma}_j^{(m)}$, $j\neq i$. Finally, make the following identifications in the covering space: 
	\begin{equation*}
			(v^{(m)}_i)_{c_-(v_i^{(m)})}^k \simeq (v_i^{(m)})_{c_+(v_i^{(m)})}^{k+\rho_i},
	\end{equation*}
	where $c_\pm (v_i^{(m)})$ stand for the two itineraries of $v_i^{(m)}$ along the cycle ${\gamma}_i^{(m)}$, cf. \eqref{eq:sg_gen_cuts}. The resultant space  $\tilde{K}\doteq K_\times/\simeq$ is comprised of sheets $K^k$, $k\in\mathbb{Z},$ and has associated graphs $\Gamma_m^k$ with vertices $V_m^k$. As before  $\Gamma^0_m$
 approximates the fundamental domain.
	
 The remainder of the procedure is implemented in analogy with the treatment of HMs on SG.
 Define the discrete spaces with appropriate boundary and jump conditions:
$$
\begin{array}{ll}
	H_m^k=\left\{ \right. f\in L(V_m^k,\R): &
	{f(v_1^k) = k} \\
	&{f(v_i^k) = f(v_{i-1}^k) + \delta_{i-1},\quad 2\leq i \leq \ell}\\
	& f((v_i^{(m)})_{{c+}}^k)=f((v_i^{(m)})_{c-}^k) +\rho_{i},
	\; i\in[ n_m] \left. \right\},
\end{array}
$$
and energy form
\begin{equation}\label{eq:energy_min_gen}
	E[\Gamma_m^k](f) = \frac{1}{2}\sum c_{ij}^m (f(u_j)-f(u_i))^2.
\end{equation}
Assumption \ref{as.harmonic} allows us to repeat the energy minimization along the lines of Lemma \ref{lem.minimizer}, resulting in a unique harmonic 
function on $\Gamma_{m+1}^0$. This can be harmonically extended as in Lemma \ref{lem.ext_harmonic_on_G} to a
dense set of $\tilde{K}$, and further extended by uniform continuity to obtain a harmonic function on the entire
$\tilde{K}$. Restricting to the fundamental domain and projecting the range to $\mathbb{T}$
yields the desired harmonic map. 
\end{proof}

\subsection{Examples}\label{sec.pcf-examples}

The following examples are meant to illustrate key steps in the construction of the covering
space for a set of representative p.c.f. fractals. With the covering space in hand, the rest
of the algorithm follows as discussed in Theorem \ref{thm.pcf}.

For each of the following examples, we use a simple method for generating the required basis
$\mathcal{B}_m$ and associated cut points. Begin with a spanning tree, $\Gamma_T$ of $\Gamma_m$.
Adding any edge, $e$, not contained in $\Gamma_T$ generates a cycle; the collection of all such cycles forms
a basis $\mathcal{B}_m$ \cite{diestel2024graph}. The cut points are constructed using the same edges, $e$.
For each edge, simply select any vertex along the embedded image of $e$ in $\Gamma_{m+1}$.

 \subsection{$\mathbf{SG_n}$}
 
 The level-$n$ SG generalizes the SG \cite{Kig01}. 
 Take $V_0 = \{v_1,v_2,v_3\}$ to be the vertices of a triangle $T$. 
 The $SG_n$ is constructed via an iterated function system by defining $n(n+1)/2$ maps
 \begin{equation*}
 	F_i(x) = \frac{1}{n}(x-c_i)+c_i, \;\; i\in [n(n+1)/2],
 \end{equation*}
 where
 $$
 c_i=\displaystyle\frac{1}{n-1}\sum_{1\leq {j_1}\leq \cdots \leq {j_{n-1}}\leq 3} v_{j_i}.
 $$ 
 Taking $n=2$ results again in SG. For $SG_3$ the $c_i$ are generated by the expressions
 $$
 \frac{1}{2}(v_j+v_k),\; 1\leq j\leq k \leq 3.
 $$


 Figure  \ref{fig:sg3_tree} shows  $\Gamma_1$ and a corresponding spanning tree for $SG_3$. 
 Figure \ref{fig:sg3_basis} shows the associated cycle basis and locations of cut points. 
 
 In contrast to $G$, the boundaries $\{\partial F_w(T)\}$ no longer form a basis for the cycle space of $\Gamma_m$. Indeed, we see $|\{\partial F_w(T): |w|\leq 1\}| = 7$, while for $\Gamma_1$ the cycle space is in fact $9$-dimensional. 
 HMs at this level.


 \subsection{Polygaskets}
 
 Polygaskets generalize the construction of SG to regular polygons. Consider a regular $n$-gon, $P$, with 
 $n\in\N$ not divisible by 4. Then define $n$ homotheties $F_i$ with fixed points the vertices of $P$ such that $F_i(P)$ and $F_j(P)$ intersect at a single point when $i\neq j$. 
 As discussed in \cite{Str06}, one can compose the $F_i$ with rotations so that it is sufficient to take as the boundary just 3 vertices of the $n$-gon: $V_0 = \{v_1,v_2,v_3\}$. 
 
 Figures \ref{fig:hexa_tree} and \ref{fig:hexa_basis} demonstrate the construction
 of the cycle basis and cut points for the first level of the $n=6$ hexagasket
 (this fractal was also studied in \cite{Tang11} following the method of \cite{Strich02}). Like the $SG_n$, edges in the hexagasket embed naturally from $\Gamma_m$ to $\Gamma_{m+1}$, making the identification of cycles and cut points in $\Gamma_{m+1}$ relatively straightforward.  
 
 In contrast, Figures \ref{fig:penta_tree} and \ref{fig:penta_cuts} show the $n=5$ pentagasket. In this case, there is a non-trivial embedding of cycles from $\Gamma_m$ to $\Gamma_{m+1}$.  Nonetheless, it is still possible to construct a cycle basis and cut points in $\Gamma_{m+1}$ satisfying
 Assumption \ref{as.cycle}.

  \begin{figure}[h]
 	\centering
 	\includegraphics[width = .45\textwidth]{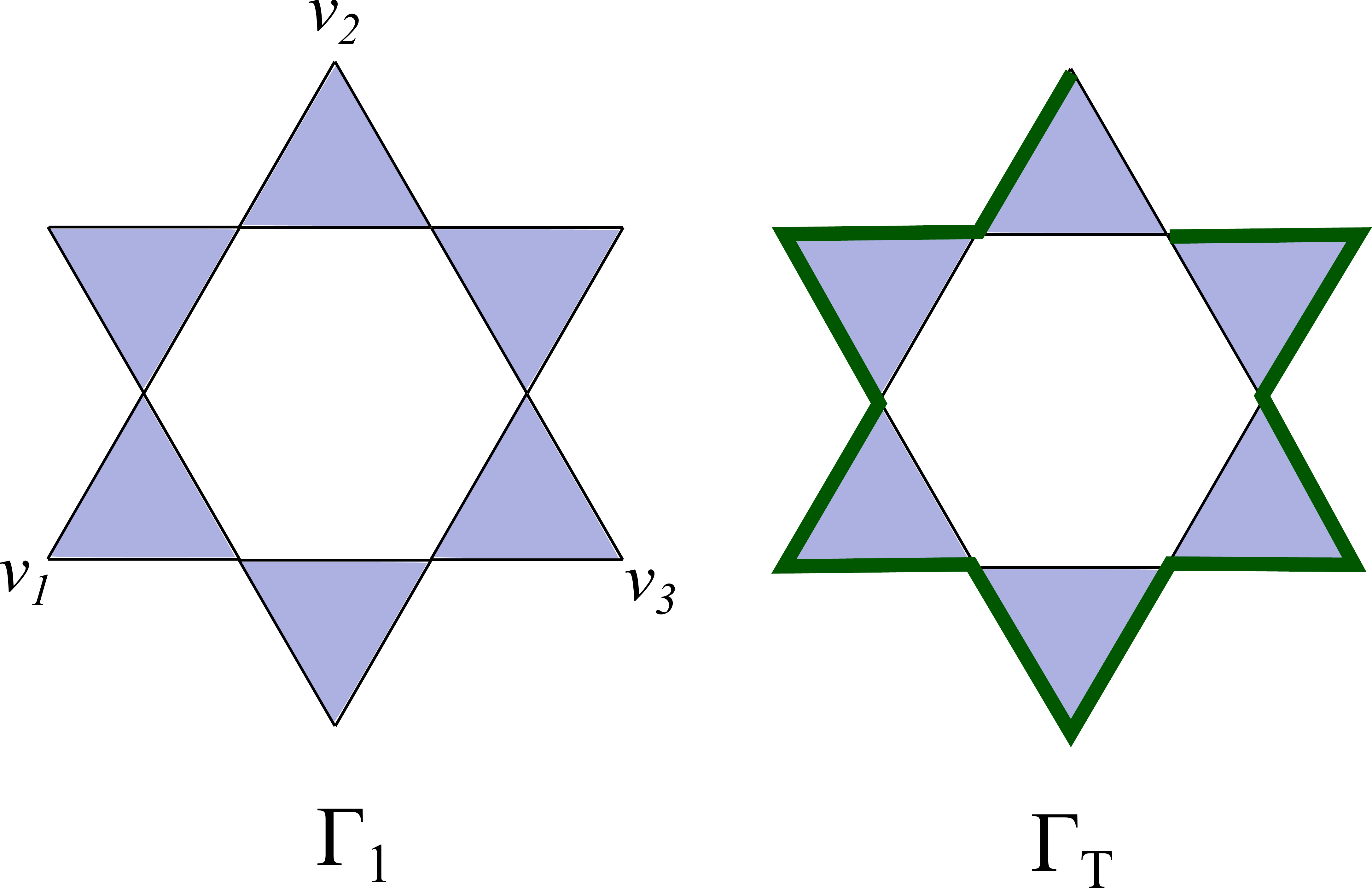}
 	\caption{Graph approximation $\Gamma_1$ for the hexagasket, and a spanning tree, $\Gamma_T$.}
 	\label{fig:hexa_tree}
 \end{figure}
 
 \begin{figure}[h]
 	\centering
 
 	\includegraphics[width = .8\textwidth]{Fig/hexa_cutpoints1.pdf}
 	\caption{Cycle basis generated by $\Gamma_T$ in Figure \ref{fig:hexa_tree}. On the right, cycles are embedded in $\Gamma_{2}$ with cut points in red.}
 	\label{fig:hexa_basis}
 \end{figure}
 
   \begin{figure}[h]
 	\centering
\includegraphics[width = .45\textwidth]{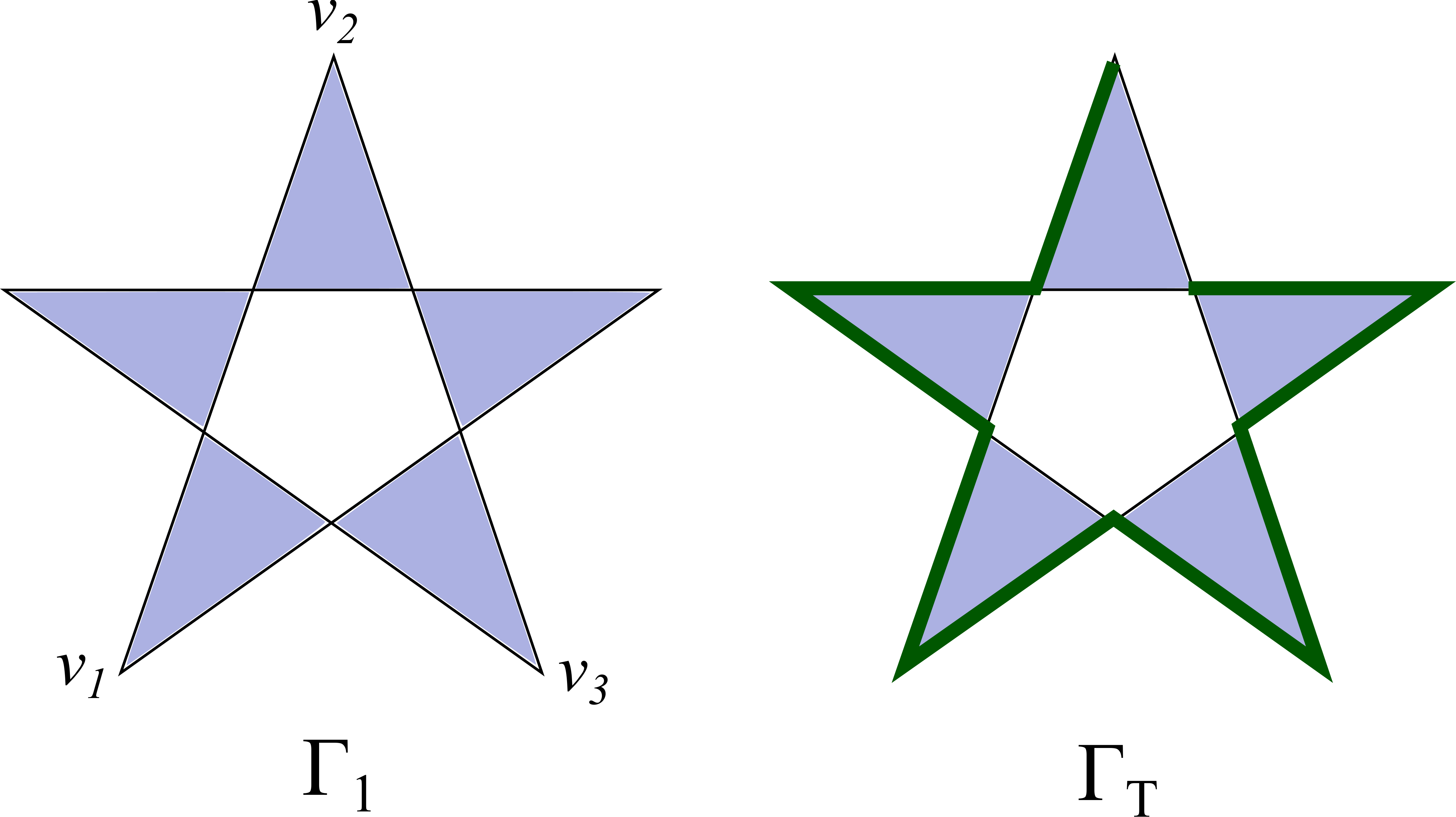}
   	\caption{Graph approximation $\Gamma_1$ for the pentagasket, and a  spanning tree, $\Gamma_T$.}
   	 	\label{fig:penta_tree}
  \end{figure} 

  \begin{figure}[h] 
  	\centering 
  	\includegraphics[width = .8\textwidth]{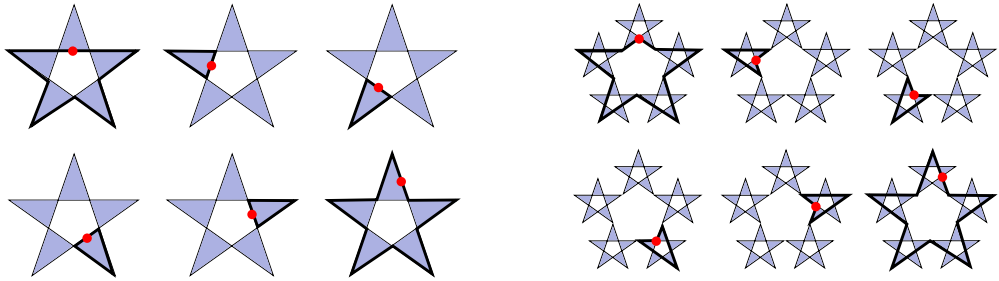}
 	\caption{Cycle basis generated by $\Gamma_T$ in Figure \ref{fig:penta_tree}. On the right, cycles are embedded in $\Gamma_{2}$ with cut points in red.}
 	\label{fig:penta_cuts}
 \end{figure}


%

\section{Discussion}
\setcounter{equation}{0}

In this paper, we presented a method for constructing HMs from the p.c.f. fractals 
to the unit circle using the covering spaces for the underlying fractal. The method
provides a complete description of HMs from the SG to the unit circle. Specifically,
it shows that there is a unique HM satisfying given boundary conditions in each homotopy class.
In addition, the method yields explicit formulae for computing HMs on 
a dense set of points. Our method may be viewed as a natural generalization of the 
classical harmonic extension algorithm for real-valued functions (cf.~\cite{Kig01}; see also \cite{Str06, Saenz-HarmAnal}).

The method extends naturally to a large class of fractals, the p.c.f. fractals.
While the implementation of the computational algorithm and the proof of the main
result remain practically the same as in the case of SG, the analysis of the HMs
on p.c.f. sets reveals subtle differences. Namely, for SG the cycle spaces of approximating 
graphs embed nicely from one level approximation to the next one. This affords an
especially nice and simple way of describing the degree of a HM on SG
(see \eqref{homotopy-c}). Examples of level-$3$ SG
and polygaskets show that in general the relation between cycles
spaces of approximating graphs at different levels of approximation may 
be more complex. This prevented us from claiming
uniqueness in Theorem~\ref{thm.pcf}, even though the gist of the
matter remained the same as before.

By construction, the HM obtained
in the proof of Theorem~\ref{thm.pcf} is still unique in a given homotopy class, provided the latter is fully captured by
\eqref{eq:winding_vec_general}, i.e., there are no nonzero winding numbers at 
smaller scales. For SG, this condition was conveniently expressed
by the nonzero entries of \eqref{homotopy-c} or by the order of the HM. While the meaning of the order of HM is intuitively clear
for a general p.c.f. set, expressing it formally meets certain
technical difficulties, which are postponed for a future study.
Nonetheless, the covering spaces for p.c.f. fractals that we introduced in this work provide an effective visual tool
for describing the 
geometry of HMs to the unit circle and we believe that they
will find other interesting applications.

\noindent {\bf Conflict of Interest Statement.} The authors declear
that there are no conflicts of interest regarding the publication of this paper.

\noindent {\bf Acknowledgements.}
This work was partially supported by the National Science Foundation through
grants DMS-2406941 (GSM) and DMS-2406942 (MM). The authors are
grateful
to Anatolii Grinshpan and Steffen Marcus
for helpful discussions.

\vfill
\newpage
\begin{appendix}
  \section{Appendix: Proof of Theorem~\ref{thm.homotopy}} \label{sec.proof}

    Suppose $f\sim g$. Then there exists $F\in C\left([0,1]\times G, \T\right)$ satisfying \eqref{F-hom}.
    Let $\gamma\in \mathcal{P}$ be arbitrary but fixed. Choose a parametrization $c_\gamma:\T\to\gamma\subset G$ and
    consider $f_\gamma=f\circ c_\gamma$ and $g_\gamma=g\circ c_\gamma$. Then $f_\gamma$ and $g_\gamma$ are two continuous maps
    from $\T$ to itself. They are homotopic with the homotopy provided
    by
    $[0,1]\times \T\ni (t,s) \mapsto F(t, c(s))$.
    By the Hopf degree theorem, $\omega(f_\gamma)=\omega(g_\gamma)$. Since $\gamma$ is arbitrary, we have
    $\bar\omega(f)=\bar\omega(g)$.

    Conversely, suppose $\bar\omega(f)=\bar\omega(g)$. In particular,
    \begin{equation}\label{for-all-gamma}
      \omega(f_\gamma)=\omega(g_\gamma)\qquad \forall \gamma\in\mathcal{P}.
    \end{equation}
We want to show that $f$ and $g$ are homotopic. This will follow from
the following two lemmas.
For convenience, in this proof we view $\T\subset \C$ as a
multiplicative group, i.e., we represent $t\in \T$ by a complex number $t = e^{i2\pi \theta}$.

\begin{lemma}\label{lem.zero}
Let $h\in C(G,\T)$ and suppose $\bar\omega(h)=(0,0,\dots)$. Then $h$ is
homotopic to $1$.
\end{lemma}
\begin{proof}
  Below we show that there is a continuous function
  $\phi\in C(G,\R)$ 
  such that
  \begin{equation}\label{h-lift}
h(x)=e^{\iu 2\pi\phi(x)}.
    \end{equation}
 
  To construct $\phi$, we first restrict $h$ to $\gamma_0$. Then
 \begin{equation}\label{arg}
\phi(x)=\frac{1}{2\pi}\arg h(x),
  \end{equation}
  where fix a branch of $\arg$ by insisting that $\phi(\xi^0)\in [0,1).$
  Recall that $\xi^k$ stands for the reference point assigned to $\gamma_k$,
  $k=0,1,2,\dots$ (see Figure~\ref{fig:f-cuts}).

The function $\phi_{\gamma_0}\doteq \phi\circ c_{\gamma_0}$
  is continuous on $[0,1)$, because $h$ is continuous on $G$. Furthermore,
$\phi_{\gamma_0}(1-0)=\phi_{\gamma_0}(0),$ because
$\omega_{\gamma_0}(h)=0$. Thus, we found a continuous function $\phi$ on
$\mathcal{P}_0$ that satisfies \eqref{h-lift} on $\mathcal{P}_0$.

Function $\phi$ is extended as a continuous function for the rest of $\mathcal{P}$ by induction.
Suppose it has been continuously defined on $\cP_m, \; m\in\N,$ then for each triangular loop
$\gamma\in \cP_{m+1}\setminus\cP_m$, it is already defined on two sides of $\gamma$ and
thus can be extended by continuity via \eqref{arg}. Since $\omega_\gamma(h)=0$,
$\phi_{\gamma}(1-0)=\phi_{\gamma}(0)$ and $\phi$ is continuous on $\gamma$.
By induction, we obtain a continuous function $\phi$ on $\cP$ that satisfies \eqref{h-lift}.

Since $h$ is continuous on $G$ and $G$ is compact, $h$ is uniformly continuous on $G$
and, consequently, $h$ is uniformly continuous on $\cP\subset G$. We want to show that
$\phi$ is uniformly continuous on $\cP$ as well, so that it can be extended from the dense
subset $\cP$ to the rest of $G$.

The uniform continuity of $h$ and \eqref{arg} imply 
\begin{equation}\label{uni-h}
  \forall 0<\epsilon<\frac{1}{2} \; \exists \delta>0:\; |x-y|<\delta  \;\implies
  d_\T\left(e^{\iu 2\pi \phi(x)}, e^{\iu 2\pi \phi(y)}\right)<\epsilon.
  \end{equation}
  Here and in the remainder of the proof, $x,y\in\cP$ and $|\cdot|$ and $d_\T(\cdot,\cdot)$
  stand for the norm in $\R^2$ and geodesic distance on $\T$ respectively.
  \begin{figure}[h]
	\centering
	\includegraphics[width  =.3\textwidth]{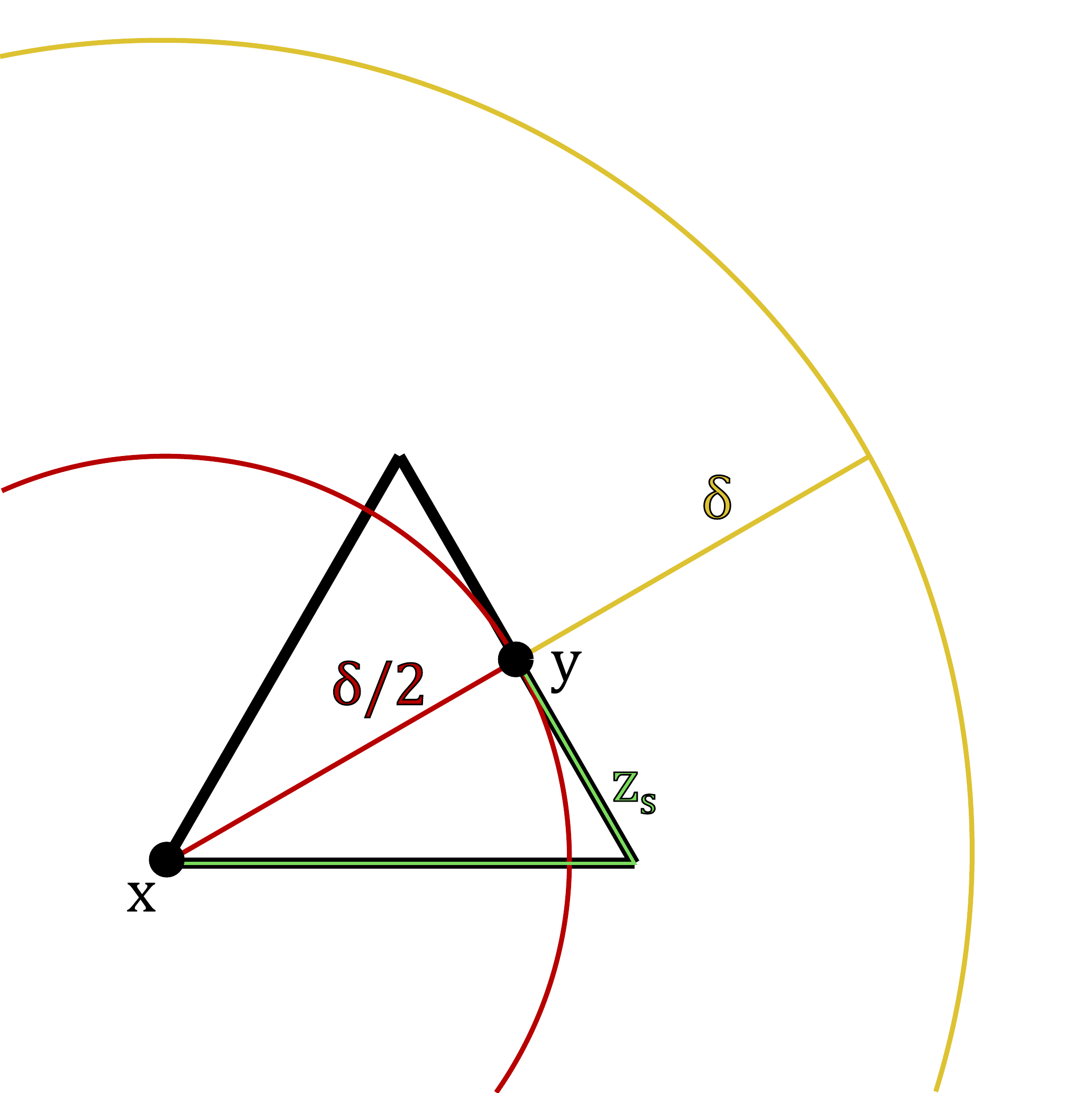}
	\caption{ Two points $x,y\in\cP$ which are at most $\delta/2$
          apart are connected by a path belonging to a $\delta$-ball
          centered at $x$.}
	\label{fig:new-fig}
\end{figure}

  From \eqref{uni-h}  to conclude that $\phi$ is uniformly continuous
  we need to check that $\phi(x)$ and $\phi(y)$ belong to the same branch of $\arg$. This follows
  from the following observation. By elementary geometry, $\forall y\in B_{\delta/2}(x)\cap \cP$ there is a
  continuous path connecting $x,y$: $(z_s,\; s\in [0,1]:\; z_0=x, \; z_1=y)$ such that
  $z_s\in B_{\delta}(x)\cap \cP$ for every $s\in [0,1]$ (see Figure~\ref{fig:new-fig}).

  By \eqref{uni-h}, $|\phi(z_s)-\phi(x)|<\epsilon$ for every $s\in [0,1]$ and, therefore,
  $|\phi(y)-\phi(x)|<\epsilon$ for every $y\in B_{\delta/2}(x)\cap \cP$.
 This shows uniform continuity of $\phi$.
Thus, $\phi$ can be extended by continuity from $\cP$ to $G$.

With \eqref{h-lift} in hand, $F(t,x)=e^{\iu 2\pi\phi(x) t}$ provides the desired homotopy.
\end{proof}

\begin{lemma}\label{lem.f-g}
  Let $f, g\in C(G,\T)$ and $\bar\omega(f)=\bar\omega(g)$. Then $\frac{f}{g}\in C(G,\T)$ and
  $\bar\omega\left(\frac{f}{g}\right)=(0,0,\dots)$.
\end{lemma}
\begin{proof}
  We are going to show that
  on $\cP$, $f$ and $g$ can be written as
  \begin{equation}\label{rep-f-g}
f(x)=e^{\iu 2\pi \phi(x)},\quad g(x)=e^{\iu 2\pi\psi(x)},\quad x\in \cP,
    \end{equation}
    where $\phi_\gamma$ and $\psi_\gamma$ are continuous for every $\gamma$ such that
    $\omega_\gamma(f)=\omega_\gamma(g)=0$. If
    $\omega_\gamma(f)\neq 0$ then
    $\phi_\gamma$ and $\psi_\gamma$ have  a jump discontinuity at a certain point on $\gamma$.
    The size of the jump is the same for both functions and is equal to $\omega_\gamma(f)$.

    The statement of the lemma follows immediately from \eqref{rep-f-g}. Indeed, 
\begin{equation}\label{f/g-lift}
\frac{f}{g}=e^{\iu 2\pi \left(\phi(x)-\psi(x)\right)},
\end{equation}
where $\phi-\psi$ is continuous on  $\cP$, because the jumps are the same for both functions
$\phi$ and $\psi$.
From \eqref{f/g-lift}, we conclude that $\bar\omega\left(\frac{f}{g}\right)=(0,0,\dots)$.

It remains to prove \eqref{rep-f-g}. We prove the representation in \eqref{rep-f-g} for $f$. It works the same
way for $g$. The proof follows the lines of the proof of Lemma~\ref{lem.zero}.

Restrict $f$ to $\gamma_0$. Let $\xi_0$ denote the point in the middle of the base of the triangular
loop $\gamma_0$ (see Figure~\ref{fig:f-cuts}). Choose
$$
\phi(x)=\frac{1}{2\pi}\arg f(x),
$$
where $\arg f(x)$ is an arbitrary fixed branch of the argument. Note that $\arg f(x)$ is continuous along $\gamma_0$
if $\omega_{\gamma_0}=0$ and it undegoes a jump at $\xi_0$ equal to $\omega_{\gamma_0}$
if $\omega_{\gamma_0}\neq 0.$

We proceed by induction. Suppose $\phi$ has been continuously defined on $\cP_N$ then
for each triangular loop $\gamma\in \cP_{N+1}\setminus\cP_N$ it is already defined on two sides
of $\gamma$.  On the remaining side of $\gamma$, we choose a point
$\xi$ in the middle of that side so that $\xi\notin \cP_N$, e.g.,
the reference point $\xi^k$ (see Figure \ref{fig:f-cuts}). If  $\omega_\gamma(f)=0$
then we extend $\phi$ to the rest of $\gamma$ by continuity to obtain a continuous function on $\gamma$.
Otherwise, we let $\phi$ to have a jump at $\xi$. Clearly, the size of the jump is equal to $\omega_\gamma(f)$.
\end{proof}

\section{Appendix: Cut points for SG}\label{appendix}
\setcounter{equation}{0}
\begin{enumerate}
\item
  As before we cut every $G^k, k\in\Z,$ at $z^k$:
  $$
z^k_-=v^k_{1\bar{3}}\quad \mbox{and} \quad z^k_+=v^k_{3\bar{1}}
  $$
  and identify  $z_-^k=z_+^{k+\rho_0}$.
\item
  In addition, for every $1\le m\le N$ and $w\in S^m$.
\begin{enumerate}
\item Denote
 \begin{align*}
  x^k_w&\doteq v^k_{w13\bar{2}} =v^k_{w12\bar{3}},\\
 y^k_w&\doteq v^k_{w23\bar{1}}=v^k_{w21\bar{3}},\\
  z^k_w&\doteq v^k_{w32\bar{1}}=v^k_{w31\bar{2}}.
  \end{align*}
\item Cut at $x_w^k, y_w^k,$ and $z_w^k$:
  \begin{align*}
  x^k_{w,+}=v^k_{w13\bar{2}}, & \quad  x^k_{w,-}  =v^k_{w12\bar{3}},\\
 y^k_{w,+}= v^k_{w23\bar{1}},  & \quad  y^k_{w,-}  =v^k_{w21\bar{3}},\\
  z^k_{w,+}= v^k_{w32\bar{1}}, &  \quad  z^k_{w,-}  =v^k_{w31\bar{2}}.
  \end{align*}
\item Identify
  $$
  x^k_{w,+}= x^{k+\rho_{\ell(w1)}}_{w,-}, \quad
 y^k_{w,+}= y^{k+\rho_{\ell(w2)}}_{w,-}, \quad
  z^k_{w,+}= z^{k+\rho_{\ell(w3)}}_{w,-},
  $$
  where
  $$
\ell(ws)= \sum_{i=1}^m w_i\cdot 3^{i-1} + s\cdot 3^m.
  $$
  \end{enumerate}
\end{enumerate}
  
\end{appendix}

\bibliographystyle{amsplain}

\begin{thebibliography}{10}

\bibitem{Braides-beginners}
Andrea Braides, \emph{{{\(\Gamma\)}}-convergence for beginners}, Oxf. Lect.
  Ser. Math. Appl., vol.~22, Oxford: Oxford University Press, 2002 (English).

\bibitem{Chi15}
Hayato Chiba, \emph{A proof of the {K}uramoto conjecture for a bifurcation
  structure of the infinite-dimensional {K}uramoto model}, Ergodic Theory
  Dynam. Systems \textbf{35} (2015), no.~3, 762--834. 

\bibitem{ChiMed19a}
Hayato Chiba and Georgi~S. Medvedev, \emph{The mean field analysis of the
  {K}uramoto model on graphs {I}. {T}he mean field equation and transition
  point formulas}, Discrete Contin. Dyn. Syst. \textbf{39} (2019), no.~1,
  131--155. 

\bibitem{ChiMed19b}
\bysame, \emph{The mean field analysis of the {K}uramoto model on graphs {II}.
  {A}symptotic stability of the incoherent state, center manifold reduction,
  and bifurcations}, Discrete Contin. Dyn. Syst. \textbf{39} (2019), no.~7,
  3897--3921. 

\bibitem{ChiMed22}
\bysame, \emph{Stability and bifurcation of mixing in the {K}uramoto model with
  inertia}, SIAM J. Math. Anal. \textbf{54} (2022), no.~2, 1797--1819.
  

\bibitem{CMM18}
Hayato Chiba, Georgi~S. Medvedev, and Matthew~S. Mizuhara, \emph{Bifurcations
  in the {K}uramoto model on graphs}, Chaos \textbf{28} (2018), no.~7, 073109,
  10. 

\bibitem{CMM2022}
\bysame, \emph{Instability of mixing in the {Kuramoto} model: from bifurcations
  to patterns}, Pure Appl. Funct. Anal. \textbf{7} (2022), no.~4, 1159--1172
  (English).

\bibitem{CMM23}
\bysame, \emph{Bifurcations and patterns in the {Kuramoto} model with inertia},
  J. Nonlinear Sci. \textbf{33} (2023), no.~5, 21 (English), Id/No 78.

\bibitem{CirGro2025}
Francisco Cirelli, Pablo Groisman, Ruojun Huang, and Hern{\'a}n Vivas,
  \emph{Scaling limit of the {Kuramoto} model on random geometric graphs}, SIAM
  J. Appl. Math. \textbf{85} (2025), no.~4, 1719--1748 (English).

\bibitem{diestel2024graph}
Reinhard Diestel, \emph{Graph theory}, Springer (print edition); Reinhard
  Diestel (eBooks), 2024.

\bibitem{DillKum02}
Stephen Dill, Ravi Kumar, Kevin~S. Mccurley, Sridhar Rajagopalan, D.~Sivakumar,
  and Andrew Tomkins, \emph{Self-similarity in the web}, ACM Trans. Internet
  Technol. \textbf{2} (2002), no.~3, 205–223.

\bibitem{DorBul12}
F.~Dorfler and F.~Bullo, \emph{Synchronization and transient stability in power
  networks and non-uniform {K}uramoto oscillators}, SICON \textbf{50} (2012),
  no.~3, 1616--1642.

\bibitem{EelSam64}
James Eells, Jr. and J.~H. Sampson, \emph{Harmonic mappings of {R}iemannian
  manifolds}, Amer. J. Math. \textbf{86} (1964), 109--160. 

\bibitem{Falc-FracGeom}
Kenneth Falconer, \emph{Fractal geometry}, third ed., John Wiley \& Sons, Ltd.,
  Chichester, 2014, Mathematical foundations and applications. 

\bibitem{Harris}
K.D. Harris and T.D. Mrsic-Flogel, \emph{Cortical connectivity and sensory
  coding}, Nature \textbf{503(7474)} (2013).

\bibitem{Jost-Riemann}
J{\"u}rgen Jost, \emph{Riemannian geometry and geometric analysis}, 7th edition
  ed., Universitext, Cham: Springer, 2017 (English).

\bibitem{Kig01}
Jun Kigami, \emph{Analysis on fractals}, Cambridge Tracts in Mathematics, vol.
  143, Cambridge University Press, Cambridge, 2001. 

\bibitem{Kur84}
Yoshiki Kuramoto, \emph{Cooperative dynamics of oscillator community}, Progress of
  Theor. Physics Supplement (1984), 223--240.

\bibitem{Kur75}
Yoshiki Kuramoto, \emph{Self-entrainment of a population of coupled non-linear
  oscillators}, International {S}ymposium on {M}athematical {P}roblems in
  {T}heoretical {P}hysics ({K}yoto {U}niv., {K}yoto, 1975), Springer, Berlin,
  1975, pp.~420--422. Lecture Notes in Phys., 39. \MR{0676492}

\bibitem{Med14a}
Georgi~S. Medvedev, \emph{The nonlinear heat equation on dense graphs and graph
  limits}, SIAM J. Math. Anal. \textbf{46} (2014), no.~4, 2743--2766.
  \MR{3238494}

\bibitem{Med2026}
\bysame, \emph{Interacting dynamical systems on networks and fractals: discrete
  and continuous models, mean-field limit, and convergence rates}, Preprint,
  {arXiv}:2601.23175 [math.{DS}] (2026), 2026.

\bibitem{MedMiz2025c}
Georgi~S. Medvedev and Matthew~S. Mizuhara, \emph{The {Kuramoto} model on the
  {Sierpinski} {Gasket} {III}: The continuum limit}, in preparation.

\bibitem{MM22}
\bysame, \emph{Chimeras unfolded}, J. Stat. Phys. \textbf{186} (2022), no.~3,
  Paper No. 46, 19. 

\bibitem{MedMiz2025}
\bysame, \emph{The {Kuramoto} model on the {Sierpinski} {Gasket} {II}: Twisted
  states}, Preprint, {arXiv}:2506.12940 [math.{MP}] (2025), 2025.

\bibitem{KurPikRos}
Arkady Pikovsky, Michael Rosenblum, and J\"{u}rgen Kurths,
  \emph{Synchronization}, Cambridge Nonlinear Science Series, vol.~12,
  Cambridge University Press, Cambridge, 2001, A universal concept in nonlinear
  sciences. 

\bibitem{Saenz-HarmAnal}
Ricardo~A. S\'aenz, \emph{Introduction to harmonic analysis}, Student
  Mathematical Library, vol. 105, American Mathematical Society, Providence,
  RI; Institute for Advanced Study (IAS), Princeton, NJ, [2023] \copyright2023,
  IAS/Park City Mathematical Subseries. 

\bibitem{Shepard}
Gordon~M. Shepherd, \emph{The {S}ynaptic {O}rganization of the {B}rain}, Oxford
  University Press, 2004.

\bibitem{Strich02}
Robert~S. Strichartz, \emph{Harmonic mappings of the {S}ierpinski gasket to the
  circle}, Proc. Amer. Math. Soc. \textbf{130} (2002), no.~3, 805--817.
  

\bibitem{Str06}
\bysame, \emph{Differential equations on fractals}, Princeton University Press,
  Princeton, NJ, 2006, A tutorial. 

\bibitem{Str-Sync}
Steven Strogatz, \emph{Sync}, Hyperion Books, New York, 2003, How order emerges
  from chaos in the universe, nature, and daily life. 

\bibitem{Str00}
Steven~H. Strogatz, \emph{From {K}uramoto to {C}rawford: exploring the onset of
  synchronization in populations of coupled oscillators}, Phys. D \textbf{143}
  (2000), no.~1-4, 1--20, Bifurcations, patterns and symmetry. 

\bibitem{StrMir91}
Steven~H. Strogatz and Renato~E. Mirollo, \emph{Stability of incoherence in a
  population of coupled oscillators}, J. Statist. Phys. \textbf{63} (1991),
  no.~3-4, 613--635.

\bibitem{Tang11}
Donglei Tang, \emph{Harmonic mappings of the hexagasket to the circle}, Anal.
  Theory Appl. \textbf{27} (2011), no.~4, 377--386. 

\bibitem{Tso}
Gilbert~C.D. Ts{\'o}, D.Y. and T.N. Wiesel, \emph{Relationships between
  horizontal interactions and functional architecture in cat striate cortex as
  revealed by cross-correlation analysis},  \textbf{6(4)} (1986), 1160--70.

\end{thebibliography}
\def\cprime{$'$} \def\cprime{$'$} \def\cprime{$'$}
\providecommand{\bysame}{\leavevmode\hbox to3em{\hrulefill}\thinspace}
\providecommand{\MR}{\relax\ifhmode\unskip\space\fi MR }
\providecommand{\MRhref}[2]{%
  \href{http://www.ams.org/mathscinet-getitem?mr=#1}{#2}
}
\providecommand{\href}[2]{#2}

\end{document}